\documentclass[12pt,a4paper]{amsart}
\usepackage{comment}
\usepackage{amsmath,amsfonts,amssymb,extarrows}
\usepackage{mathtools}
\usepackage{color}
\usepackage{xcolor}
\usepackage[hmargin=3cm,vmargin=3cm]{geometry}
\usepackage{hyperref}
\usepackage{nameref,zref-xr}                    
\usepackage[all]{xy}           
\usepackage{cancel}
\usepackage{longtable}
\usepackage{dsfont}
\usepackage{multicol}
\usepackage{bbold}
\numberwithin{equation}{section}
\newtheorem{theorem}{Theorem}[section]   
\newtheorem{definition}[theorem]{Definition}
\newtheorem{proposition}[theorem]{Proposition}
\newtheorem{lemma}[theorem]{Lemma}
\newtheorem{corollary}[theorem]{Corollary}

\newtheorem{example-notation}[theorem]{Example-Notation}
\newtheorem{remark}[theorem]{Remark}

\def\d{\partial}

\def\f{\frac}

\def\inw{\in\{1,\dots,n\}}

\newcommand{\eqa}{\begin{eqnarray}}
\newcommand{\eeqa}{\end{eqnarray}}
\newcommand{\beq}{\begin{equation}}
\newcommand{\eeq}{\end{equation}}

\allowdisplaybreaks

\begin{document}
\title[Generalised hodograph method for non-diagonalisable systems]{The generalised hodograph method for non-diagonalisable integrable systems of hydrodynamic type}
\author{Paolo Lorenzoni}
\address{P.~Lorenzoni:\newline Dipartimento di Matematica e Applicazioni, Universit\`a di Milano-Bicocca, \newline
Via Roberto Cozzi 55, I-20125 Milano, Italy and INFN sezione di Milano-Bicocca}
\email{paolo.lorenzoni@unimib.it}
\author{Sara Perletti}
\address{S.~Perletti:\newline Dipartimento di Matematica e Applicazioni, Universit\`a di Milano-Bicocca, \newline
Via Roberto Cozzi 55, I-20125 Milano, Italy and INFN sezione di Milano-Bicocca}
\email{sara.perletti1@unimib.it}
\author{Karoline van Gemst}
\address{K.~van Gemst:\newline Dipartimento di Matematica e Applicazioni, Universit\`a di Milano-Bicocca, \newline
Via Roberto Cozzi 55, I-20125 Milano, Italy and INFN sezione di Milano-Bicocca}
\email{karoline.vangemst@unimib.it}

\begin{abstract}
\sloppypar{We extend the generalised hodograph method to regular non-diagonalisable integrable systems of hydrodynamic type, in light of the relation between such systems and F-manifolds with compatible connection. The method allows  the construction of solutions starting from  the  symmetries of the  system. In the diagonal case, the completeness of the symmetries follows from the integrability conditions that  ensure  the applicability of a Darboux's theorem on Pfaffian systems. In the regular non-diagonalisable case the validity of this  theorem relies on some further assumptions that we discuss in detail. Under these assumptions, the method provides the general solution as in Tsarev's diagonal case.}
\end{abstract}    

\maketitle

\tableofcontents

\section*{Introduction}
\textcolor{white}{...}
\vspace{-2.4em}
\newline
\newline
The present paper is devoted to extending the generalised hodograph method to non-diagonalisable systems of hydrodynamic type. Such systems appear naturally in a range of different contexts. By virtue of the techniques provided by the geometry of F-manifolds, the complexity typically portrayed by integrability in the non-diagonalisable setting is here partially tamed.
\newline
\newline
The generalised hodograph method was introduced by Tsarev in \cite{ts91} as a technique to solve diagonalisable integrable systems of hydrodynamic type:
\begin{equation*}
	\textbf{u}_t=V(\textbf{u})\,\textbf{u}_x.
\end{equation*}
The diagonalisability of such a system entails the existence of a set of coordinates, called \emph{Riemann invariants}, by means of which the system can be reduced to the diagonal form
\begin{equation}\label{dsht}
	\begin{bmatrix}
		u^{1}_t\cr
		u^2_t\cr
		\vdots \cr
		u^{n}_t
	\end{bmatrix}=
	\begin{bmatrix}
		v^{1} & 0 & \dots & 0\cr
		0  & v^{2}  & \dots & 0\cr
		\vdots & \ddots & \ddots & \vdots\cr
		0 & \dots & 0 & v^{n}
	\end{bmatrix}\begin{bmatrix}
		u^{1}_x\cr
		u^2_x\cr
		\vdots \cr
		u^{n}_x
	\end{bmatrix}.
\end{equation}
The functions $\{v^{i}(u^1,...,u^n)\}_{i\inw}$ are called \emph{characteristic velocities}, and are assumed to be pairwise distinct. By considering a vector field $X$ having such characteristic velocities as components in the Riemann invariants $X^i=v^i$, $i\inw$, and by introducing the commutative and associative product $\circ$ defined by the structure constants
\[c^i_{jk}=\delta^i_j\delta^i_k,\qquad i,j,k\inw,\]
the system \eqref{dsht} can be written in the equivalent form
\beq\label{Fsys}
u^i_t=c^i_{jk}X^ju^k_x,\qquad i\inw.
\eeq
The generalised hodograph method allows one to construct solutions
\begin{equation*}
{\bf  u}(x,t)=(u^1(x,t),...,u^n(x,t))
\end{equation*}
of the system \eqref{dsht}, in the implicit form
\begin{equation}\label{hm}
	x+v^i({\bf u}(x,t))t=w^i({\bf u}(x,t)),\qquad i\inw,
\end{equation}
in terms of the characteristic velocities $\{w^{i}(u^1,...,u^n)\}_{i\inw}$ of the symmetries,
\begin{equation}\label{symdsht}
	\begin{bmatrix}
		u^{1}_\tau\cr
		u^2_\tau\cr
		\vdots \cr
		u^{n}_\tau
	\end{bmatrix}=
	\begin{bmatrix}
		w^{1} & 0 & \dots & 0\cr
		0  & w^{2}  & \dots & 0\cr
		\vdots & \ddots & \ddots & \vdots\cr
		0 & \dots & 0 & w^{n}
	\end{bmatrix}\begin{bmatrix}
		u^{1}_x\cr
		u^2_x\cr
		\vdots \cr
		u^{n}_x
	\end{bmatrix},
\end{equation}
of the system \eqref{dsht}. The existence of symmetries relies on Tsarev's integrability 
condition
\begin{equation}
	\label{tsarev1}
	\partial_j\Gamma^i_{ik}=\partial_k\Gamma^i_{ij},\qquad i\ne j\ne k\ne i,
\end{equation}
for the quantities
\begin{equation}\label{ChS}
	\Gamma^i_{ij}:=\frac{\partial_j v^i}{v^j-v^i},\qquad i\ne j.
\end{equation}
This integrability condition may be equivalently formulated as
\begin{equation}\label{tsarev2}
	\d_i\Gamma^k_{kj}+\Gamma^k_{ki}\Gamma^k_{kj}-\Gamma^k_{kj}\Gamma^j_{ji}
	-\Gamma^k_{ki}\Gamma^i_{ij}=0,
\end{equation} 
due to the identity
\begin{equation}\label{tsarevid}
	\d_i\Gamma^k_{kj}+\Gamma^k_{ki}\Gamma^k_{kj}-\Gamma^k_{kj}\Gamma^j_{ji}
	-\Gamma^k_{ki}\Gamma^i_{ij}=\f{v^i-v^k}{v^j-v^i}(\d_j\Gamma^k_{ki} -\d_i\Gamma^k_{kj}),
\end{equation}
proved in \cite{ts91}. Conditions \eqref{tsarev1} and \eqref{tsarev2} ensure the existence of a family of symmetries \eqref{symdsht}, parametrised by $n$ arbitrary functions of a single variable, one for each Riemann invariant of the system.
\newline
\newline
\sloppy
In the present paper, we extend the generalised hodograph method to non-diagonalisable systems of hydrodynamic type of the form
\beq\label{SHT-intro}
u^i_t=V^i_j(u)u^j_x,\qquad i\inw,
\eeq
defined by a block-diagonal matrix $V=\text{diag}(V_{(1)},...,V_{(r)})$, $r\leq n$, whose generic $\alpha^{\text{th}}$ block, of size $m_\alpha$, is of the lower-triangular Toeplitz form
\begin{equation}\label{toeplitz}
	V_{(\alpha)}=
	\begin{bmatrix}
		v^{1(\alpha)} & 0 & \dots & 0\cr
		v^{2(\alpha)} & v^{1(\alpha)} & \dots & 0\cr
		\vdots & \ddots & \ddots & \vdots\cr
		v^{m_\alpha(\alpha)} & \dots & v^{2(\alpha)} & v^{1(\alpha)}
	\end{bmatrix}.
\end{equation}
Analogously to the semisimple setting, such block-diagonal systems can be written in the form \eqref{Fsys}, by introducing the commutative and associative product $\circ$ defined by the structure constants
\begin{equation}\label{strconsts_reg}
	c^{i(\alpha)}_{j(\beta)k(\gamma)}=\delta^\alpha_\beta\delta^\alpha_\gamma\delta^i_{j+k-1},
\end{equation}
for all $\alpha,\beta,\gamma\in\{1,\dots,r\}$ and  $i\in\{1,\dots,m_\alpha\}$, $j\in\{1,\dots,m_\beta\}$, ${k\in\{1,\dots,m_\gamma\}}$. These products appear in the theory of F-manifolds with Euler vector fields. More precisely, coordinates realising \eqref{Fsys}, called \emph{canonical coordinates}, were provided in \cite{DH} when the product defined by \eqref{strconsts_reg} induces the structure of an F-manifold with Euler vector field and satisfies the \emph{regularity} requirement. That is, the operator of multiplication by the Euler vector field admits a canonical decomposition such that to each Jordan block there corresponds a distinct eigenvalue. It will be  convenient to relabel the coordinates $(u^1,\dots,u^n)$ according to the following rule. For each $\alpha\in\{2,\dots,r\}$ and for each $j\in\{1,\dots,m_\alpha\}$ we write
\begin{equation}
	\label{relabellingcoordinates}
	j(\alpha)=m_1+\dots+m_{\alpha-1}+j,
\end{equation}
so that $u^{j(\alpha)}$ denotes the $j$-th coordinate associated to the Jordan block with label $\alpha$ (for $\alpha=1$ we set $j(\alpha)=j$). The regular setting reduces to the semisimple one when $r=n$, in which case one retrieves systems of hydrodynamic type admitting a diagonal representation.
\newline
\newline
Integrable systems of hydrodynamic type appear in particular as related to a class of F-manifolds introduced in \cite{LPR} (see also \cite{LP}) under the name of \emph{F-manifolds with compatible connection (and flat unit)}. These were extensively studied in \cite{LPVG}. This relation is pivotal in generalising Tsarev's approach to the non-diagonalisable setting, as in the semisimple case the quantities \eqref{ChS} can be identified with part of the Christoffel symbols of a torsionless  connection  $\nabla$ (called in \cite{LP} the \emph{natural connection}) and Tsarev's  integrability conditions coincide with the vanishing of some components of the associated Riemann tensor. As shown in \cite{LPR}, Tsarev's condition \eqref{tsarev2} can be written in one of the two equivalent forms
\newline
\beq\label{shc-intro}
R^s_{lmi}c^j_{ks}+R^s_{lik}c^j_{ms}+R^s_{lkm}c^j_{is}=0,\qquad R^j_{skl}c^s_{mi}+R^j_{smk}c^s_{li}+R^j_{slm}c^s_{ki}=0.
\eeq
\newline
In \cite{LPVG}, it was proved that, under the regularity assumption, the existence of  $n$ commuting flows of the form \eqref{Fsys} allows one to  introduce a connection $\nabla$ satisfying
the conditions \eqref{shc-intro}. As a consequence, in the regular case,  an integrable system of hydrodynamic type locally defines the structure of an F-manifold with compatible connection (and flat unit). By exploiting this relation, it is possible to generalise Tsarev's integrability method to the non-semisimple regular setting. In accordance with the double-index notation \eqref{relabellingcoordinates}, let us introduce the $n$-dimensional vector  
\newline
\[(v^1,\dots,v^n)=(v^{1(1)},\dots,v^{m_1(1)},\dots,v^{1(r)},\dots,v^{m_r(r)}),\] 
\newline
where  $n=m_1+\cdots+m_r$. A key point is the following.
\newline
\newline
{\bf Main observation I}. \emph{The matrix with entries
\begin{equation*}
	M^i_j=t\,\d_jv^i-\d_jw^i,\qquad i,j\inw,
\end{equation*}
evaluated  at the solution ${\bf  u}(x,t)$ of  the algebraic system
\[x\,\delta^{i}_1+t\,v^{i(\alpha)}=w^{i(\alpha)},\qquad\alpha\in\{1,...,r\},\,i\in\{1,...,m_{\alpha}\},\]
has precisely the same form as the matrix $V$ defining the system \eqref{SHT-intro}, meaning that there exists a vector valued function $Z$ such that
 $M({\bf  u}(x,t))=Z\circ$. For instance, in the regular  case and in canonical coordinates, $M({\bf  u}(x,t))$ has a block-diagonal form, each block being as in \eqref{toeplitz}.}
\newline
\newline
Using this crucial observation, it is not difficult to prove that the solution of the same algebraic system is also a solution of the system \eqref{SHT-intro}.
\newline
\newline
In the diagonal Tsarev's case, the linear system of  PDEs for the symmetries
 has the form studied by Darboux in \emph{Leçons sur les systemes orthogonaux et les coordonées curvilignes} and the integrability condition \eqref{shc-intro} 
 ensures the applicability of Theorem III (p. 335 of \cite{darboux}). Tsarev used a slight  extension of this  theorem (Proposition 1 in \cite{ts91}) to prove that, in a neighbourhood of  a point $(x_0,t_0)$, any solution can be obtained applying the generalised hodograph method. 
\newline
\newline 
The non-diagonal case is more involved. Another key point is the following.
\newline
\newline
{\bf Main observation II}. 
\emph{In the non-diagonal regular case which we  study in this paper, Darboux's Theorem III cannot be directly applied, since the linear  system for the symmetries does not in general have  the form considered by Darboux. This requires some further assumptions. Under these additional assumptions,  in a neighbourhood of  a point, any solution can be obtained applying the generalised hodograph method. } 
\newline
\newline
In the case of a single Jordan block, for instance, where
 \begin{equation*}
	V=
	\begin{bmatrix}
		v^{1} & 0 & \dots & 0\cr
		v^{2}  & v^{1} & \dots & 0\cr
		\vdots & \ddots & \ddots & \vdots\cr
		v^{n} & \dots & v^{2} & v^{1}
	\end{bmatrix},
\end{equation*}
one has to assume that $v^i=v^i(u^1,\dots,u^i),\,i\inw$.
\newline
\newline
The extension of the generalised hodograph method to non-diagonalisable systems of hydrodynamic type finds application in many areas. Among them, we mention generalised Lenard-Magri chains associated with the Frölicher–Nijenhuis bicomplex, non-semisimple (bi-)flat F-manifolds and Dubrovin-Frobenius manifolds (see, for instance, \cite{LPR,LP,LP22,LP23} and references therein) and Hamiltonian and bihamiltonian structures of hydrodynamic type and their deformations (see \cite{DVLS}). It is also worth mentioning \cite{KK,KO}. Some generalisation of the  hodograph method appeared in literature in the study of some specific examples (for instance in the study of reductions of the soliton gas kinetic equation \cite{FP,VF}) but as far as we know  a general proof of this method has never been given before. Remarkably the case of a single Jordan block of arbitrary size is related to the mKP hierarchy (see \cite{XF}). We think that
the understanding of integrability in the non-diagonalisable case will be  also important in  the  study  of dispersive  deformations of systems of  hydrodynamic type  \cite{DLZ}.  The general results obtained  so far require the existence of Riemann invariants for the dispersionless limit. Without this assumption only a few preliminary results are  available (see  for  instance \cite{DVLS}  for  the bihamiltonian  case).
\newline
\newline
The paper is organised  as  follows. In Section \ref{SectionFmnfs}, we recall the definition of an F-manifold, the notion of regularity, and the additional structure induced by a compatible connection. In Section \ref{SectionGenHodMethod}, we extend the generalised hodograph method for such systems, adapting the original construction provided by Tsarev in \cite{ts91}. In order to do this it is necessary to fully exploit the correspondence between such  systems  and F-manifolds with compatible connection. Section \ref{section:int} consists of a discussion of the  symmetries for non-diagonalisable integrable systems of hydrodynamic type admitting a block-diagonal Toeplitz form as introduced above. In particular, we consider the compatibility of the system of symmetries in the wake of the results obtained in \cite{LPVG} in relation with integrable hierarchies. Section \ref{section:completeness} is devoted to a further study of the symmetries. Here, we prove necessary and sufficient conditions for the applicability of Darboux's Theorem III  to the linear systems for the symmetries. We call  this property \emph{completeness}. This implies that the general solution of the system depends on $n$ arbitrary functions of a single variable. The main difference with respect to the diagonal case is that not  all  the variables appear in these arbitrary functions:   each block contributes with a number of functions equal to the size of the block and depending  on the \textit{main variable} of the block, i.e. ${u^{1(\alpha)}}$ for the $\alpha^{\text{th}}$ block. Remarkably it turns out that, as in Tsarev's diagonal case, the completeness of the  symmetries allows one to obtain any solution using the generalised hodograph method. 
In the final section, we draw conclusions and, motivated by some examples, speculate a relation between completeness of the symmetries and the existence of Hamiltonian structures.
\newline
\newline
\noindent{\bf  Data availability}. No datasets were generated or analysed in this work.
\newline
\newline
\noindent{\bf Acknowledgements}. The authors are supported by funds of  INFN   (Istituto Nazionale di Fisica Nucleare) by IS-CSN4 Mathematical Methods of Nonlinear Physics. Authors are also thankful to GNFM (Gruppo Nazionale di Fisica Matematica) for supporting activities that contributed to the research reported in this paper. This research has received funding by the Italian PRIN 2022 (2022TEB52W)  \emph{The charm of integrability: from nonlinear waves to random matrices}.

\newpage
\section{Preliminaries}
\subsection{Regular F-manifolds with compatible connection}\label{SectionFmnfs}
\subsubsection{F-manifolds}
F-manifolds were introduced by Hertling and Manin in \cite{HM} and are defined as follows.
\begin{definition}\label{defFmani}
An \emph{F-manifold} is a manifold $M$ equipped with
\begin{itemize}
\item[(i)] a commutative associative bilinear product  $\circ$  on the module of (local) vector fields, satisfying the following identity:
\begin{equation}\label{HM}
\mathcal{L}_{X\circ Y} \circ=X\circ (\mathcal{L}_Y \circ) +Y\circ (\mathcal{L}_X\circ ),
\end{equation}
for all local vector fields $X,Y$;
\newline
\item[(ii)] a distinguished vector field $e$ on $M$ such that 
\[e\circ X=X\] 
for all local vector fields $X$.
\end{itemize}
\end{definition}
\textcolor{white}{...}
\vspace{-2.4em}
\newline
\newline
Condition \eqref{HM} is known as the \emph{Hertling-Manin condition} and can be shown to be equivalent to
\begin{align}
&[X\circ Y,W\circ Z]-[X\circ Y, Z]\circ W-[X\circ Y, W]\circ Z\label{HMeq1free}\\
&-X\circ [Y, Z \circ W]+X\circ [Y, Z]\circ W +X\circ [Y, W]\circ Z\notag\\
&-Y\circ [X,Z\circ W]+Y\circ [X,Z]\circ W+Y\circ [X, W]\circ Z=0,\notag
\end{align}
for all local vector fields $X,Y,W, Z$, where $[X,Y]$ is the Lie bracket.
\subsubsection{F-manifolds with compatible connection}
F-manifolds are usually equipped with additional structures. 
\begin{definition}\label{defFwithE}
An \emph{F-manifold with Euler vector field} is an F-manifold $M$ equipped with a vector field $E$ satisfying
\begin{equation}
	\mathcal{L}_E \circ=\circ.\label{Euler}
\end{equation}
\end{definition}
\textcolor{white}{...}
\vspace{-2.4em}
\newline
\newline
One can easily observe that \eqref{Euler} implies that
\begin{equation*}
	[e,E]=e.
\end{equation*}
Following \cite{LPR}, we now introduce the notion of F-manifold with compatible connection (and flat unit).
\begin{definition}\label{Fmnfwcompatconn_flatunit}
An \emph{F-manifold with compatible connection and flat unit} is a manifold $M$ equipped with a product 
\[\circ : TM \times TM \rightarrow TM,\] 
with structure functions $c^i_{jk}$, a connection $\nabla$ with Christoffel symbols 
$\Gamma^i_{jk}$ and a distinguished vector field $e$ such that
\begin{itemize}
\item[(i)] the one-parameter family of  connections $\{\nabla_{\lambda}\}_{\lambda}$ with Christoffel symbols
$$\Gamma^i_{jk}-\lambda c^i_{jk},$$
gives a torsionless connection for any choice of $\lambda$ for which the Riemann  tensor coincides
 with the Riemann tensor of $\nabla \equiv \nabla_0$
\beq\label{curvlambda}
R_{\lambda}(X,Y)(Z)=R(X,Y)(Z),
\eeq
and satisfies the condition
\beq\label{rc-intri}
Z\circ R(W,Y)(X)+W\circ R(Y,Z)(X)+Y\circ R(Z,W)(X)=0,
\eeq
for all local vector fields $X$, $Y$, $Z$, $W$;
\item[(ii)] $e$ is the unit of the product;
\item[(iii)] $e$ is flat: $\nabla e=0$.
\end{itemize}
\end{definition}
\textcolor{white}{...}
\vspace{-2.4em}
\newline
\newline
Let us discuss some consequences of condition \eqref{curvlambda}.  
For a given $\lambda$, the torsion and curvature are respectively given by
\begin{eqnarray*}
T^{(\lambda)k}_{ij}&=&\Gamma^k_{ij}-\Gamma^k_{ji}+\lambda(c^k_{ij}-c^k_{ji}),\\
R^{(\lambda)k}_{ijl}&=&R^k_{ijl}+\lambda(\nabla_i c^k_{jl}-\nabla_j c^k_{il})+\lambda^2(c^k_{im}c^m_{jl}-c^k_{jm}c^m_{il}),
\end{eqnarray*}
where $R^k_{ijl}$ is the Riemann tensor of $\nabla$. Thus, condition (i) is equivalent to
\begin{enumerate}
\item the vanishing of the torsion of $\nabla$;
\item the commutativity of the product  $\circ$;
\item the symmetry  in the lower indices of the tensor field $\nabla_l c^k_{ij}$;
\item the associativity of the product $\circ$.
\end{enumerate}

\begin{remark}
Since the Hertling-Manin condition \eqref{HMeq1free} follows from the symmetry  in the lower indices of the tensor field $\nabla_l c^k_{ij}$, as shown in \cite{He02},  F-manifolds with compatible connection and flat unit constitute a broad class of F-manifolds.
\end{remark}

\begin{remark}
The  operator  of multiplication by the Euler vector field has vanishing Nijenhuis torsion:
$$N_{E\circ}(X,Y):=[E\circ X,E\circ Y]+E\circ E\circ [X,Y]-E\circ [X,E\circ Y]-E\circ [E\circ X,Y]=0.$$
By virtue of this property, families of examples were constructed in \cite{LP23,LPVG}.
\end{remark}

\begin{remark}	
Condition \eqref{rc-intri} can be written in the equivalent  form  
\begin{equation}\label{rc-intri-2}
R(Y,Z)(X\circ W)+R(X,Y)(Z\circ W)+R(Z,X)(Y\circ W)=0,
\end{equation}
for all local vector fields $X$, $Y$, $Z$, $W$ (see \cite{LPR} for  details).
\end{remark}
\textcolor{white}{...}
\vspace{-2.4em}
\newline
\newline
In local coordinates conditions \eqref{rc-intri} and \eqref{rc-intri-2} read  \eqref{shc-intro}.
\subsubsection{Regularity} 
Following \cite{DH}, we present the notion of a regular F-manifold.
\begin{definition}[\cite{DH}]\label{DavidHertlingdef}
An F-manifold with Euler vector field $(M,\circ, e,E)$ is called \emph{regular} if for each $p\in M$ the matrix representing the endomorphism
$$L_p := E_p\circ : T_pM \to T_pM$$
has exactly one Jordan block for each distinct eigenvalue.
\end{definition}
\textcolor{white}{...}
\vspace{-2.4em}
\newline
\newline
In this paper, we will use the following important result of \cite{DH} regarding the existence of non-semisimple canonical coordinates for regular F-manifolds with Euler vector field.
\begin{theorem}[\cite{DH}]\label{DavidHertlingth}
Let $(M, \circ, e, E)$ be a regular F-manifold of dimension $n \geq 2$ with an Euler vector field $E$. Furthermore, assume that locally around a point $p\in M$, the Jordan canonical form of the operator $L$ has $r$ Jordan blocks of sizes $m_1,...,m_r$ with distinct eigenvalues. Then
there exists locally around $p$ a distinguished system of coordinates $\{u^1, \dots, u^{m_1+\dots+m_r}\}$ such that
 \begin{align}
	e^{i(\alpha)}&=\delta^i_1,\qquad
	E^{i(\alpha)}=u^{i(\alpha)},\qquad
	c^{i(\alpha)}_{j(\beta)k(\gamma)}=\delta^\alpha_\beta\delta^\alpha_\gamma\delta^i_{j+k-1},\notag
\end{align}
for all $\alpha,\beta,\gamma\in\{1,\dots,r\}$ and  $i\in\{1,\dots,m_\alpha\}$, $j\in\{1,\dots,m_\beta\}$, ${k\in\{1,\dots,m_\gamma\}}$.
\end{theorem}
\textcolor{white}{...}
\vspace{-2.4em}
\newline
\newline
The coordinates defining this special system are called \emph{David-Hertling coordinates}, and in such coordinates the operator  of multiplication by the Euler vector field takes the lower-triangular Toeplitz form \eqref{toeplitz} mentioned in the introduction  with $v^{i(\alpha)}=u^{i(\alpha)}$. We will also often refer to David-Hertling coordinates as \emph{canonical coordinates}.
\subsubsection{Integrable hierarchies and  F-manifolds with  compatible connection}\label{SubsectionIntHsAndFmnfsCC}
According to the results of \cite{LPVG}, given  an $n$-dimensional regular F-manifold with Euler vector field $(M,\circ,e,E)$ and $n$ commuting flows of the form
\begin{equation}
u^i_t=c^i_{jk}X^j_{(s)}u^k_x,\qquad i,s\inw,
\end{equation} 
defined by a frame of vector fields $(X_{(1)},...,X_{(n)})$, under mild technical assumptions, a unique torsionless connection $\nabla$ is defined by the conditions
\begin{equation}\label{cond1}
\nabla_j e^i=0,\qquad i,j\inw,
\end{equation}
\begin{equation}\label{cond2}
(d_{\nabla}V)^i_{jk}=0,\qquad i,j,k\inw,
\end{equation}
where $V=X\circ$, defines on $M$ the structure of an F-manifold with compatible connection and flat unit. More precisely, we recall the following.
\begin{proposition}\label{LPVG_Prop3.10}(\cite{LPVG})
	For $(M,\circ,e)$ being an $n$-dimensional regular F-manifold and $X$ being a local vector field realising
	\begin{itemize}
		\item $X^{1(\alpha)}\neq X^{1(\beta)}$ for $\alpha\neq\beta$, $\alpha,\beta\in\{1,\dots,r\}$,
		\item $X^{2(\alpha)}\neq0$, $\alpha\in\{1,\dots,r\}$,
	\end{itemize}
	there exists a unique torsionless connection $\nabla$ such that $d_\nabla(X\circ)=0$ and $\nabla e=0$.
\end{proposition}

\begin{remark}\label{ChristoffelExplicit}
	The proof of Proposition \ref{LPVG_Prop3.10} provides explicit formulas for the Christoffel symbols $\{\Gamma^i_{jk}\}_{i,j,k\in\{1,\dots,n\}}$ of $\nabla$, in terms of $\{\partial_j X^i\}_{i,j\in\{1,\dots,n\}}$, which we shall now recall. Let us fix $\alpha,\beta,\gamma\in\{1,\dots,r\}$, $\alpha\neq\beta\neq\gamma\neq\alpha$, and $i\in\{1,\dots,m_\alpha\}$. The Christoffel symbols of the unique torsionless connection satisfying $\nabla e=0$ and $d_\nabla V=0$, for $V=X\circ$, are determined by the following formulas:\small
	\begin{align}
		 \quad&\Gamma^{i(\alpha)}_{j(\beta)k(\gamma)}=0,\qquad j\in\{1,\dots,m_\beta\},\,k\in\{1,\dots,m_\gamma\};\label{Chr_a}\\
		 \label{Chr_b}\quad&\Gamma^{i(\alpha)}_{j(\beta)k(\alpha)}=-\frac{1}{X^{1(\alpha)}-X^{1(\beta)}}\Bigg(\partial_{j(\beta)}X^{(i-k+1)(\alpha)}+\overset{m_\alpha}{\underset{s=k+1}{\sum}}\Gamma^{i(\alpha)}_{j(\beta)s(\alpha)}X^{(s-k+1)(\alpha)}\\&\quad\quad\qquad\,\,\,\,-\overset{m_\beta}{\underset{s=j+1}{\sum}}\Gamma^{i(\alpha)}_{k(\alpha)s(\beta)}X^{(s-j+1)(\beta)}\Bigg), \qquad k\in\{1,\dots,m_\alpha\},\,j\in\{1,\dots,m_\beta\};\notag\\
		\quad&\Gamma^{i(\alpha)}_{j(\beta)1(\beta)}=-\Gamma^{i(\alpha)}_{j(\beta)1(\alpha)}\label{Chr_nablae_1},\qquad j\in\{1,\dots,m_\beta\},\\
		&\Gamma^{i(\alpha)}_{j(\alpha)1(\alpha)}=-\overset{}{\underset{\sigma\neq\alpha}{\sum}}\,\Gamma^{i(\alpha)}_{j(\alpha)1(\sigma)},\qquad j\in\{1,\dots,m_\alpha\};\label{Chr_nablae_2}\\
		\quad&\Gamma^{i(\alpha)}_{k(\alpha)j(\alpha)}=\frac{1}{X^{2(\alpha)}}\bigg(\partial_{(j-1)(\alpha)}X^{(i-k+1)(\alpha)}+\overset{m_\alpha}{\underset{s=k+1}{\sum}}\Gamma^{i(\alpha)}_{(j-1)(\alpha)s(\alpha)}X^{(s-k+1)(\alpha)}\label{Chr_d}\\&\quad\quad\quad\,\,\,\,-\partial_{k(\alpha)}X^{(i-j+2)(\alpha)}-\overset{m_\alpha}{\underset{s=j+1}{\sum}}\Gamma^{i(\alpha)}_{k(\alpha)s(\alpha)}X^{(s-j+2)(\alpha)}\bigg),\notag\\&\text{for } j, k\in\{2,\dots,m_\alpha\}, \text{ where } j \leq k;\notag\\
		&\Gamma^{i(\alpha)}_{k(\beta)j(\beta)}=\frac{1}{X^{2(\beta)}}\Bigg(\overset{m_\beta}{\underset{s=k+1}{\sum}}\Gamma^{i(\alpha)}_{(j-1)(\beta)s(\beta)}X^{(s-k+1)(\beta)}-\overset{m_\beta}{\underset{s=j+1}{\sum}}\Gamma^{i(\alpha)}_{k(\beta)s(\beta)}X^{(s-j+2)(\beta)}\Bigg),\label{Chr_c}\\&\text{for }k, j\in\{2,\dots,m_\beta\}, \text{ where } j \leq k.\notag
	\end{align}\normalsize
\end{remark}
\textcolor{white}{...}
\vspace{-2.4em}
\begin{remark}\label{rmkSym}
As observed in \cite{LPVG}, thanks to the identity $\nabla_ec^i_{jk}=0$, for any local vector field $Y$, the condition 
\[(d_{\nabla}(Y\circ))^i_{jk}=c^i_{js}\nabla_kY^s-c^i_{ks}\nabla_jY^s+(\nabla_kc^i_{js}-\nabla_jc^i_{ks})Y^s=0\]
 implies separately the condition
\begin{equation}\label{symndshtY2}
c^i_{js}\nabla_kY^s-c^i_{ks}\nabla_jY^s=0
\end{equation}
and the condition
\begin{equation}\label{symndshtY3}
(\nabla_kc^i_{js}-\nabla_jc^i_{ks})Y^s=0.
\end{equation}
\end{remark}
An F-manifold with compatible connection and flat unit can be constructed by means of such a connection, when the local vector field $X$ meeting the assumptions of Proposition \ref{LPVG_Prop3.10} belongs to a set of $n$ linearly independent local vector fields defining pairwise commuting flows.
\begin{theorem}\label{LPVG_mainTh}(\cite{LPVG})
	Let $\{X_{(0)},\dots,X_{(n-1)}\}$ be a set of linearly independent local vector fields on an $n$-dimensional regular F-manifold $(M,\circ,e)$. Let us assume that the corresponding flows
	\begin{align}
		{\bf  u}_{t_i}=X_{(i)}\circ {\bf u}_x,\qquad i\in\{0,\dots,n-1\},
		\notag
	\end{align}
	pairwise commute, and that there exists a local vector field $X\in\{X_{(0)},\dots,X_{(n-1)}\}$ such that ${X^{1(\alpha)}\neq X^{1(\beta)}}$ for ${\alpha\neq\beta}$, $\alpha,\beta\in\{1,\dots,r\}$, and $X^{2(\alpha)}\neq0$ for each ${\alpha\in\{1,\dots,r\}}$. Then an F-manifold $(M,\circ,e,\nabla)$ with compatible connection and flat unit is defined by the unique torsionless connection $\nabla$ realising $\nabla e=0$ and $d_\nabla(X\circ)=0$.
\end{theorem}
\textcolor{white}{...}
\vspace{-2.4em}
\newline
\newline
In particular, the condition $d_\nabla(X\circ)=0$ implies that $d_\nabla(Y\circ)=0$ for any other local vector field $Y$ defining a flow which commutes with the flow defined by $X$ (in particular, $d_\nabla(X\circ)=0$ guarantees that $d_\nabla(Y\circ)=0$ for any other $Y\in\{X_{(0)},\dots,X_{(n-1)}\}$). The proof that an F-manifold with compatible connection and flat unit can be so defined is based on the fact that if $Z$ is a local vector field realising $d_\nabla(Z\circ)=0$ then
\begin{equation}\label{LPVG_eq: 3RCX}
	(R^k_{lmi}c^h_{pk} + R^k_{lip}c^h_{mk} + R^k_{lpm}c^h_{ik})Z^l = 0,\qquad i,p,m,h\inw,
\end{equation}
for any $Z\in\{X_{(0)},\dots,X_{(n-1)}\}$. The assumption of $\{X_{(i)}\}_{i\in\{0,\dots,n-1\}}$ providing a basis allows one to conclude that \eqref{LPVG_eq: 3RCX} must hold for any arbitrary local vector field $Z$, implying \eqref{shc-intro}.
\subsection{Symmetries}\label{SubsectionSymmetries}
We now consider symmetries for a system of hydrodynamic type of the form
\begin{equation*}
	u^i_t=V^i_j(u)u^j_x,\qquad i\inw,
\end{equation*}
defined by a block-diagonal matrix $V=\text{diag}(V_{(1)},...,V_{(r)})$, $r\leq n$, whose generic $\alpha^{\text{th}}$ block, of size $m_\alpha$, is of the lower-triangular Toeplitz form \eqref{toeplitz}. Recall that such a system can be written as
\begin{equation}\label{sys_original_sectionsymms}
	u^i_t=c^i_{jk}X^ju^k_x,\qquad i\inw,
\end{equation}
where $V=X\circ$ with respect to the commutative and associative product $\circ$, defined by the structure constants
\begin{equation*}
	c^{i(\alpha)}_{j(\beta)k(\gamma)}=\delta^\alpha_\beta\delta^\alpha_\gamma\delta^i_{j+k-1},
\end{equation*}
for all $\alpha,\beta,\gamma\in\{1,\dots,r\}$ and $i\in\{1,\dots,m_\alpha\}$, $j\in\{1,\dots,m_\beta\}$, ${k\in\{1,\dots,m_\gamma\}}$, and associated to an F-manifold structure $(\circ,e,\nabla)$ with compatible connection and flat unit. By \emph{symmetry}, we mean a vector field $Y$ defining a flow
\begin{equation}\label{sys_second_sectionsymms}
	u^i_\tau=c^i_{jk}Y^ju^k_x,\qquad i\inw,
\end{equation}
which commutes with \eqref{sys_original_sectionsymms}.
\begin{remark}
	If $d_\nabla(X\circ)=0$, then \eqref{sys_original_sectionsymms} and \eqref{sys_second_sectionsymms} define commuting flows if and only if $d_\nabla(Y\circ)=0$. We refer to \cite{LPVG} for details.
\end{remark}

\section{The generalised hodograph method}\label{SectionGenHodMethod}
\textcolor{white}{...}
\vspace{-2.4em}
\newline
\newline
We consider systems of hydrodynamic type of the form
\begin{equation*}
	\textbf{u}_t=V(\textbf{u})\,\textbf{u}_x,
\end{equation*}
for functions $\textbf{u}(x,t)=(u^1(x,t),...,u^n(x,t))$ taking values in an $n$-dimensional regular F-manifold $(M,\circ,e,E)$, where the $(1,1)$-type tensor $V$ is the operator of multiplication by some vector field $X$ on $M$. Let $r$ be the number of Jordan blocks of $E\circ$, and let $m_\alpha$ denote the size of the $\alpha^{\text{th}}$ block. In canonical coordinates, such a system reads
\begin{equation}\label{ndshtX}
	u^i_t=V^i_ku^k_x=c^i_{jk}X^ju^k_x,\qquad i\inw,
\end{equation}
where
\begin{align}
	c^{i(\alpha)}_{j(\beta)k(\gamma)}=\delta^\alpha_\beta\delta^\alpha_\gamma\delta^i_{j+k-1},
    \label{cijkmulti}
\end{align}
for all $\alpha,\beta,\gamma\in\{1,\dots,r\}$ and $i\in\{1,\dots,m_\alpha\}$, $j\in\{1,\dots,m_\beta\}$, ${k\in\{1,\dots,m_\gamma\}}$. The matrix representing $V$ assumes a block-diagonal form $V=\text{diag}(V_{(1)},...,V_{(r)})$, where the generic $\alpha^{\text{th}}$ block is of the lower-triangular Toeplitz form \eqref{toeplitz}.
\newline
\newline
In order to prove  the first main  result of the paper we need two technical  Lemmas.

\begin{lemma}\label{MToeplitz_implied}
Let $c^i_{jk}$ be the structure functions of an F-manifold. A matrix valued function $M \coloneqq (M^i_j)$ satisfying 
    \begin{equation}
c^i_{js}M_k^s=c^i_{ks}M_j^s,\qquad i,j,k\inw,
\label{eq:cond3temporary}
    \end{equation}
    has the form $M^i_j=c^i_{jk}Y^k$ for some vector $Y$. In particular, in the regular case and in  canonical coordinates,  $M$
  is a block-diagonal lower-triangular matrix of Toeplitz type, namely
	\begin{equation}\label{MToeplitz_gen}
		M^{i(\alpha)}_{j(\beta)}=\delta^\alpha_\beta\,M^{(i-j+1)(\alpha)}_{1(\alpha)}\,\mathbb{1}_{\{i\geq j\}}.
	\end{equation}
\end{lemma}
\begin{proof}
Firstly we observe that any $M$ of the form $M_k^i=c^i_{ks}Y^s$ satisfies $c^i_{js}M_k^s=c^i_{ks}M_j^s$ (by associativity).
\newline
\newline
Multiplying both sides of \eqref{eq:cond3temporary} by $e^j$ we obtain
\[M_k^i=c^i_{ks}Y^s\]  
with $Y^s=M_j^se^j$. For  instance, in  the regular case
\[M_{j(\beta)}^{i(\alpha)}=c^{i(\alpha)}_{j(\beta)s(\sigma)}Y^{s(\sigma)}=\delta^\alpha_\beta  Y^{(i+1-j)(\beta)}\,\mathbb{1}_{\{i\geq j\}}.\]  
For $\alpha\ne\beta$ we have $M_{j(\beta)}^{i(\alpha)}=0$. This means that \[Y^{i(\alpha)}=M_{s(\sigma)}^{i(\alpha)}e^{s(\sigma)}=M_{1(\alpha)}^{i(\alpha)},\]
and  therefore  
\[M_{j(\alpha)}^{i(\alpha)}=Y^{(i+1-j)(\alpha)}\,\mathbb{1}_{\{i\geq j\}}=M_{1(\alpha)}^{(i+1-j)(\alpha)}\,\mathbb{1}_{\{i\geq j\}}.\]	
\end{proof}
\begin{lemma}\label{InverseToeplitz}
Any invertible matrix valued  function of the form
$V^i_j=c^i_{js}X^s$ where $c^i_{js}$ are the structure functions of an F-manifold
has inverse $W^i_j=c^i_{jk}(X^{-1})^k$
where $Y=X^{-1}$ is the  unique solution of the system
$V^i_jY^j=e^i$. In particular, if an invertible matrix valued function is of block-diagonal lower-triangular Toeplitz type, then its inverse is as well.
\end{lemma}

\begin{proof}
The  proof is  a straightforward computation:
\[V^i_mW^m_j=c^i_{ms}X^sc^m_{jk}(X^{-1})^k=c^i_{mj}X^sc^m_{sk}(X^{-1})^k=c^i_{mj}\left(X\circ X^{-1}\right)^m=c^i_{mj}e^m=\delta^i_j.\]
Any matrix valued function of block-diagonal lower-triangular Toeplitz type can be written as $V^i_j=c^i_{js}X^s$, where $c^i_{js}$ are David-Hertling structure constants.  Conversely if  $c^i_{js}$ are David-Hertling structure constants then $V^i_j=c^i_{js}X^s$ has block-diagonal lower-triangular Toeplitz form. In this case,  from the above  computation, it follows that also the inverse $W^i_j=c^i_{jk}(X^{-1})^k$ has block-diagonal lower-triangular Toeplitz form.
\end{proof}

\vspace{.5 cm}
We can now state our first  main result about block diagonal system of hydrodynamic type
\begin{equation}\label{flowX}
{\bf u}_t=X\circ {\bf u}_x
\end{equation}
where $\circ$  defines a regular F-manifold structure and the vector field $X$ satisfies the usual non-vanishing  conditions:
$X^{1(\alpha)}\neq X^{1(\beta)}$ for $\alpha\neq\beta$, $X^{2(\alpha)}\neq0$.	
\begin{theorem}\label{thm_genHodMethod}
Let $\nabla$ be the torsionless connection uniquely defined by the conditions \eqref{cond1} and \eqref{cond2}, and let the vector-valued function ${\textbf{u}(x,t)=(u^1(x,t),...,u^n(x,t))}$ satisfy the algebraic system
\begin{equation}\label{ndhm}
x\,e^{i}+t\,X^{i}(\textbf{u}(x,t))=Y^{i}(\textbf{u}(x,t)),\qquad i\inw,
\end{equation}
where the vector field $Y$ defines a  flow
\begin{equation}\label{flowY}
{\bf u}_\tau=Y\circ {\bf u}_x
\end{equation}
commuting with \eqref{flowX}. Then, $\textbf{u}(x,t)$ is a solution of  the system \eqref{flowX}.
\end{theorem}
\begin{proof} 
By differentiating \eqref{ndhm} with respect to $x$ and $t$ respectively, we  get
\begin{equation}\label{xtder}
e^i+M^i_j(\textbf{u}(x,t))u^j_x=0,\qquad X^i+M^i_j(\textbf{u}(x,t))u^j_t=0,\qquad i\inw,
\end{equation}
where
\begin{equation}\label{TsMatrix}     
M^i_j:=t\,\d_jX^i-\d_jY^i,\qquad i,j\inw.
\end{equation}
The function $\textbf{u}(x,t)$ is  implicitly defined by the algebraic system \eqref{ndhm} and thus we can assume that $M$ is  non-degenerate
  when restricted  to $\textbf{u}(x,t)$.
In the diagonal case, Tsarev's proof revolves around the observation that, on the solution  
 $\textbf{u}(x,t)$ of the algebraic system \eqref{ndhm} defined by the vector field $Y$ which is associated with the symmetries \eqref{flowY} of \eqref{flowX}, the matrix $M$ is diagonal. By using this fact, and taking into account that in the diagonal case $e^i=1$ for each $i\inw$, equations \eqref{xtder} immediately reduce to
\begin{equation}\label{xtderbis}
1+M^i_i
(\textbf{u}(x,t))u^i_x=0,\qquad X^i+M^i_i(\textbf{u}(x,t))u^i_t=0,\qquad i\inw,
\end{equation}
implying the result.
\newline
\newline
It turns out that a similar idea proves  successful in the non-diagonal regular case as well. Firstly, observe that  on the solution  
 $\textbf{u}(x,t)$ of the algebraic system \eqref{ndhm} we have
\begin{equation}\label{TsMatrix_nabla}
M^i_j(\textbf{u}(x,t))=t\,\nabla_jX^i-\nabla_jY^i,\qquad i,j\inw,
\end{equation} 
as
\begin{align*}
	t\,\nabla_jX^i-\nabla_jY^i&=t\,\d_jX^i-\d_jY^i+\Gamma^i_{js}(t\,X^s-Y^s)\\
	&\overset{\eqref{ndhm}}{=}t\,\d_jX^i-\d_jY^i-x\,\Gamma^i_{js}e^s\\
	&\overset{\eqref{cond1}}{=}t\,\d_jX^i-\d_jY^i,
\end{align*}
for all $i,j\inw$. Moreover, by using \eqref{cond2} and 
\begin{equation}\label{symndshtY}
(d_{\nabla}(Y\circ))^i_{jk}=0,\qquad i,j,k\inw.
\end{equation}
and taking into account Remark \ref{rmkSym}, one obtains  the condition \eqref{eq:cond3temporary}. In turn, by Lemma \ref{MToeplitz_implied}, $M(\textbf{u}(x,t))$ has precisely the form $M^i_j(\textbf{u}(x,t))=c^i_{jk}Z^k$  for some vector valued function $Z$ and from Lemma \ref{InverseToeplitz} it follows that $(M(\textbf{u}(x,t))^{-1})^i_j=c^i_{jk}(Z^{-1})^k$.  Using  these  facts, we obtain from \eqref{xtder} 
\[u^k_x=-c^k_{is}(Z^{-1})^se^i=-(Z^{-1})^k,\]
and
\[u^k_t=-c^k_{is}(Z^{-1})^sX^i=c^k_{is}X^iu^s_x.\]
\end{proof}

\begin{remark}
Notice that if we start from an  F-manifold with compatible connection and flat unit $(M,\nabla,\circ,e)$ and  we assume  that the  vector field $X$ defining
 the flow \eqref{flowX} and the vector field $Y$ defining the symmetry \eqref{flowY} satisfy $d_{\nabla}(X\circ)=0$ and  $d_{\nabla}(Y\circ)=0$
 respectively, then the proof of Theorem  \ref{thm_genHodMethod} does not require the regularity assumption.
\end{remark}

\section{Integrability}
\label{section:int}
In this Section we discuss  the  compatibility  of the system for the symmetries 
\begin{equation}\label{hierarchy}
(d_{\nabla}W)^i_{jk}=c^i_{js}\nabla_kY^s-c^i_{ks}\nabla_jY^s=0,\qquad i,j,k\inw,
\end{equation}
of equation \eqref{flowX} where $\nabla$ is the unique torsionless connection determined by $\nabla e=0$ and $d_\nabla(X\circ)=0$ (whose Christoffel symbols are given in Remark \ref{ChristoffelExplicit} in canonical coordinates). By multiplying \eqref{hierarchy} by $e^j$ and summing over $j$, we get
\begin{equation}\label{red-hie}
\nabla_kY^i=c^i_{ks}\nabla_eY^s,\qquad i,k\inw,
\end{equation}
or equivalently
\begin{equation}\label{redhie}
\partial_kY^i=-\Gamma^i_{ks}Y^s+c^i_{ks}\nabla_eY^s\qquad i,k\inw.
\end{equation}

\begin{remark}
Equation \eqref{red-hie} is equivalent to  equation \eqref{hierarchy}. In fact, while it was already shown above that \eqref{hierarchy} implies \eqref{red-hie}, the converse implication can be proved by observing that
\begin{eqnarray*}
	c^i_{js}\nabla_kY^s-c^i_{ks}\nabla_jY^s=(c^i_{js}c^s_{kt}-c^i_{ks}c^s_{jt})\nabla_e Y^t=0,\qquad i,j,k\inw,
\end{eqnarray*}
by means of \eqref{red-hie} and the associativity of the product.
\end{remark}

Let us now compute the compatibility conditions, 
\[(\d_j\d_k-\d_k\d_j)Y^i=0,\]
of the  system \eqref{redhie}. We have
\begin{eqnarray}\notag
	\partial_j\partial_kY^i&=&-(\partial_j\Gamma^i_{ks})Y^s-\Gamma^i_{ks}(\partial_jY^s)+(\partial_jc^i_{ks})\nabla_eY^s+c^i_{ks}(\partial_j\nabla_eY^s)\\
	&=&-(\partial_j\Gamma^i_{ks})Y^s+\Gamma^i_{ks}(\Gamma^s_{jt}Y^t-c^s_{jt}\nabla_eY^t)+(\partial_jc^i_{ks})\nabla_eY^s+c^i_{ks}(\partial_j\nabla_eY^s),\label{Compat_temp_1}
\end{eqnarray}
implying
\begin{align*}
	(\partial_j\partial_k-\partial_k\partial_j)Y^i&=(\partial_k\Gamma^i_{js}-\partial_j\Gamma^i_{ks}+\Gamma^i_{kt}\Gamma^t_{js}-\Gamma^i_{jt}\Gamma^t_{ks})Y^s\\&
	+(\partial_jc^i_{kt}-\partial_kc^i_{jt}+\Gamma^i_{js}c^s_{kt}-\Gamma^i_{ks}c^s_{jt})\nabla_eY^t\\&+c^i_{ks}(\partial_j\nabla_eY^s)-c^i_{js}(\partial_k\nabla_eY^s)\\&=R^i_{sjk}Y^s+(\partial_jc^i_{kt}-\partial_kc^i_{jt}+\Gamma^i_{js}c^s_{kt}-\Gamma^i_{ks}c^s_{jt})\nabla_eY^t\\
	&+c^i_{ks}(\partial_j\nabla_eY^s)-c^i_{js}(\partial_k\nabla_eY^s).
\end{align*}
Taking into account the symmetry of $\nabla  c$, we obtain
\begin{align*}
	(\partial_j\partial_k-\partial_k\partial_j)Y^i&=R^i_{sjk}Y^s+(\Gamma^s_{jt}c^i_{ks}-\Gamma^s_{kt}c^i_{js})\nabla_eY^t+c^i_{ks}(\partial_j\nabla_eY^s)-c^i_{js}(\partial_k\nabla_eY^s)\\
	&=R^i_{sjk}Y^s+c^i_{ks}(\Gamma^s_{jt}\nabla_eY^t+\partial_j\nabla_eY^s)-c^i_{js}(\Gamma^s_{kt}\nabla_eY^t+\partial_k\nabla_eY^s)\\
	&=R^i_{sjk}Y^s+c^i_{ks}\nabla_j\nabla_eY^s-c^i_{js}\nabla_k\nabla_eY^s\\
	&=R^i_{sjk}Y^s+c^i_{ks}e^m\nabla_j\nabla_mY^s-c^i_{js}e^m\nabla_k\nabla_mY^s
\end{align*}
which, since $\nabla_j\nabla_mY^s=\nabla_m\nabla_jY^s+R^s_{tmj}Y^t$, becomes
\begin{align*}
	(\partial_j\partial_k-\partial_k\partial_j)Y^i&=R^i_{sjk}Y^s+c^i_{ks}e^mR^s_{tmj}Y^t-c^i_{js}e^mR^s_{tmk}Y^t\\&+c^i_{ks}e^m\nabla_m\nabla_jY^s-c^i_{js}e^m\nabla_m\nabla_kY^s.
\end{align*}
In particular, by \eqref{redhie}, $\nabla_jY^s=c^s_{jt}\nabla_eY^t$ and $\nabla_kY^s=c^s_{kt}\nabla_eY^t$, giving
\begin{align*}
	c^i_{ks}e^m\nabla_m\nabla_jY^s-c^i_{js}e^m\nabla_m\nabla_kY^s&=c^i_{ks}e^m\nabla_m(c^s_{jt}\nabla_eY^t)-c^i_{js}e^m\nabla_m(c^s_{kt}\nabla_eY^t)\\
&=c^i_{ks}(e^m\nabla_mc^s_{jt})\nabla_eY^t-c^i_{js}(e^m\nabla_mc^s_{kt})\nabla_eY^t\\
	&=c^i_{ks}(\nabla_ec^s_{jt})\nabla_eY^t-c^i_{js}(\nabla_ec^s_{kt})\nabla_eY^t,
\end{align*}
due to the associativity of the product. Since $\nabla_ec^s_{jt}=0$ as a consequence of \eqref{HMeq1free} (see \cite{LPVG} for details), we get
\begin{align*}
	c^i_{ks}e^m\nabla_m\nabla_jY^s-c^i_{js}e^m\nabla_m\nabla_kY^s&=0,
\end{align*}
yielding
\begin{align}\label{Compat_temp_2}
	(\partial_j\partial_k-\partial_k\partial_j)X^i&=R^i_{sjk}Y^s+c^i_{ks}e^mR^s_{tmj}Y^t-c^i_{js}e^mR^s_{tmk}Y^t\\
	&=(R^i_{tjk}+c^i_{ks}R^s_{tmj}e^m+c^i_{js}R^s_{tkm}e^m)Y^t\notag\\
	&=[(c^i_{sm}R^s_{tjk}+c^i_{ks}R^s_{tmj}+c^i_{js}R^s_{tkm})e^m]Y^t.\notag
\end{align}
Thus, the compatibility conditions follow from
\begin{equation}\label{red-rc}
	\left(R^s_{lmi}c^j_{ks}+R^s_{lik}c^j_{ms}+R^s_{lkm}c^j_{is}\right)e^m=0.
\end{equation}

\begin{remark}\label{RemarkCompatVS3RC}
	Condition \eqref{red-rc} amounts to condition \eqref{rc-intri-2}, or equivalently to its local formulations \eqref{shc-intro}. In fact, while recovering \eqref{red-rc} from \eqref{shc-intro} is immediate, the opposite implication can be proved by first observing that \eqref{red-rc} can be written as
	\[R^i_{tjk}+c^i_{ks}R^s_{tmj}e^m+c^i_{js}R^s_{tkm}e^m=0\]
	and by then noticing that this implies that
	\begin{align*}
		&c^i_{sm}R^s_{tjk}+c^i_{ks}R^s_{tmj}+c^i_{js}R^s_{tkm}
		\\
		&=c^i_{ms}(c^s_{kh}R^h_{tjr}-c^s_{jh}R^h_{tkr})e^r+c^i_{ks}(c^s_{jh}R^h_{tmr}-c^s_{mh}R^h_{tjr})e^r+c^i_{js}(c^s_{mh}R^h_{tkr}-c^s_{kh}R^h_{tmr})e^r\\
		&=R^h_{tjr}e^r(c^i_{ms}c^s_{kh}-c^i_{ks}c^s_{mh})
		+R^h_{tkr}e^r(c^i_{js}c^s_{mh}-c^i_{ms}c^s_{jh})+R^h_{tmr}e^r(c^i_{ks}c^s_{jh}-c^i_{js}c^s_{kh})=0,
	\end{align*}
	due to the associativity of the product. Thus, the condition of compatibility of the system \eqref{redhie} is equivalent to the condition \eqref{rc-intri-2}.
\end{remark}
In the semisimple case, the system \eqref{red-hie} for $Y$ becomes
\begin{equation*}
	\nabla_kY^i=\delta^i_{k}\nabla_eY^i,\qquad i,k\inw,
\end{equation*}
which gives
\beq\label{red-hie-ss}
\nabla_k Y^i=0,\qquad k\ne i,
\eeq
or, by taking into account that $\Gamma^i_{jk}=0$ for pairwise distinct indices and $\Gamma^i_{jj}=-\Gamma^i_{ij}$ for $i\ne j$,
\begin{equation}\label{redhie-ss}
	\partial_kY^i=-\Gamma^i_{ks}Y^s=\Gamma^i_{ik}(Y^k-Y^i),\qquad k\ne i.
\end{equation}
Notice that for $k=i$, due to \eqref{red-hie-ss}, the two sides trivially coincide. In particular, the system \eqref{redhie-ss} can be explicitly written in terms of $X$, as
\begin{equation}\label{redhie-ss_X}
	\partial_kY^i=\frac{Y^k-Y^i}{X^k-X^i}\,\partial_kX^i,\qquad k\ne i.
\end{equation}
Due to Theorem III in \cite{darboux}, the compatibility of the system \eqref{redhie-ss} implies that  the general solution of \eqref{redhie-ss}  depends on $n$  arbitrary  functions  of a single variable. Tsarev proved that this  family of solutions allows one to solve the Cauchy problem for  the system defined by $Y$  (Proposition 1 in \cite{ts91}).
\newline
\newline
In the non-diagonalisable setting, this is no longer the case. For instance, when the operator of multiplication by the Euler vector field is associated with a single Jordan block, the system \eqref{red-hie} for $Y$ reduces to an identity for $k=1$, while for $k\ne 1$ it reduces to
\begin{equation}\label{red-hie-nss}
	\partial_kY^i=-\Gamma^i_{ks}Y^s+c^i_{ks}\partial_1Y^s,
\end{equation}
since $\Gamma^s_{1t}=0$ by the flatness of the unit vector  field. In this case, due to the presence of the partial derivatives $\partial_1Y^s$, we cannot straightforwardly apply Darboux's Theorem III. Nevertheless, in Section \ref{section:completeness}, we will show that, under suitable assumptions on $X$, such a system can be carefully recomposed into a family of subsystems, to each of which Darboux's Theorem III can be applied. When this is the case, in the generic regular setting, the general solution to the system for the symmetries depends, as in the semisimple case, on $n$ arbitrary functions of a single variable. Unlike the semisimple setting, however, such arbitrary functions
\begin{equation}\label{narbryfcts_regular}
	\big\{\{f_{i(\alpha)}(u^{1(\alpha)})\}_{i\in\{1,\dots,m_\alpha\}}\big\}_{\alpha\in\{1,\dots,r\}},
\end{equation}
only depend on the main variable of each block. This behaviour lead us to formulate the following.
\begin{definition}
	We say that the set of symmetries for the system \eqref{sys_original_sectionsymms} is \emph{complete} if the generic solution to $d_\nabla(Y\circ)=0$ depends on $n$ arbitrary functions of a single variable, of the form \eqref{narbryfcts_regular}.
\end{definition}

\section{Completeness of symmetries}
\label{section:completeness}
\textcolor{white}{...}
\vspace{-2.4em}
\newline
\newline
The present section is devoted to providing necessary and sufficient conditions for the applicability of Darboux's Theorem III  to the linear systems for the symmetries.	This leads naturally to the notion of completeness of the set of symmetries. As pointed out above, such symmetries are determined as solutions to $d_\nabla(Y\circ)=0$, where $\nabla$ is uniquely expressed in terms of $X$. 
\newline
\newline
We first observe some properties of $\nabla$, involving its Christoffel symbols $\{\Gamma^i_{jk}\}$ in canonical coordinates. In particular, we notice that some vanish by construction. Then, we prove some technical results which allows one to phrase the vanishing of an additional subset of Christoffel symbols in terms of the requirement for $X$ to be independent of some canonical coordinates. More precisely, we get the condition $\partial_{i(\beta)}X^{j(\alpha)}=0$ for $i>j$. We show that the vanishing of such an additional subset of Christoffel symbols, or equivalently the independence of $X$ on some canonical coordinates, is a necessary condition in order for Darboux's theorem to be applied to the system for the symmetries. Finally, we prove that the assumption that $X$ does not depend on such subset of canonical coordinates is also sufficient to apply Darboux's theorem to the system for the symmetries. Under such an assumption, since the system for the symmetries is compatible, the set  of symmetries is complete. In this setting,  the symmetries exist and are determined in terms of $n$ functions of a single variable, in analogy with the semisimple case, with these variables being the main variables of the blocks, $\{u^{1(\alpha)}\}_{\alpha\in\{1,\dots,r\}}$. 
\newline
\newline
For the convenience of the reader, we split the discussion into two subsections, starting with $V$ consisting of a single Jordan block and then treating the case of an arbitrary number of Jordan blocks. First, however, we prove three useful lemmas. 
\begin{lemma}
Let $\alpha \neq \beta$. Then,
\begin{subequations}
\begin{equation}
  \hspace{-4.5em}  \Gamma^{i(\alpha)}_{j(\beta) k(\alpha)} = \begin{cases}
        \Gamma^{(i-k+1)(\alpha)}_{j(\beta) 1(\alpha)}, & \text{if } i \geq k,\\
        0,& \text{otherwise};
    \end{cases}
    \label{lemma91a}
    \end{equation}
    \begin{equation}
        \Gamma^{i(\alpha)}_{j(\beta)k(\beta)} = \begin{cases}
            \Gamma^{i(\alpha)}_{(j+k-1)(\beta) 1(\beta)}, & \text{if } j+k \leq m_{\beta} + 1\\
            0, & \text{otherwise}.
        \end{cases}
        \label{lemma91b}
    \end{equation}
\end{subequations}
    \label{lemma:9.1}
\end{lemma}

\begin{proof}
    Lemma \ref{lemma:9.1} follows as a consequence of the symmetry of $(\nabla c)^m_{ijk}$ in $i,j$, which reads
    \begin{equation}
        \Gamma^{m(\nu)}_{i(\alpha) s(\sigma)} \delta^{\sigma}_{\beta} \delta^{\sigma}_{\gamma} \delta^{s}_{j+k-1} - \Gamma^{s(\sigma)}_{i(\alpha) k(\gamma)} \delta^{\nu}_{\beta} \delta^{\nu}_{\sigma} \delta^{m}_{j+s-1} - \Gamma^{m(\nu)}_{j(\beta) s(\sigma)} \delta^{\sigma}_{\alpha} \delta^{\sigma}_{\gamma} \delta^{s}_{i+k-1} + \Gamma^{s(\sigma)}_{j(\beta) k(\gamma)} \delta^{\nu}_{\alpha} \delta^{\nu}_{\sigma} \delta^{m}_{i+s-1} = 0,
        \label{nablacsymcoords}
    \end{equation}
    by using \eqref{cijkmulti}. 
    \newline
    \newline
Firstly, let $\alpha = \gamma = \nu \neq \beta$, and let $k=1$. Then \eqref{nablacsymcoords} reads
\begin{equation*}
    \Gamma^{m(\alpha)}_{j(\beta) i(\alpha)} = \Gamma^{(m-i+1)(\alpha)}_{j(\beta) 1(\alpha)}.
\end{equation*}
If $m < i$ the right-hand side is undefined and we get 
\begin{equation*}
    \Gamma^{m(\alpha)}_{i(\alpha) j(\beta)} = 0, 
\end{equation*}
proving \eqref{lemma91a}. 
\newline
\newline
    Now, let $\alpha = \beta = \gamma \neq \nu$, and let $i=1$. Then \eqref{nablacsymcoords} reads
    \begin{equation*}
        \Gamma^{m(\nu)}_{1(\alpha) (j+k-1)(\alpha)} = \Gamma^{m(\nu)}_{j(\alpha) k(\alpha)}.
    \end{equation*}
    If $j+k-1>m_{\alpha}$ the left-hand side is undefined and we get
    \begin{equation*}
        \Gamma^{m(\nu)}_{j(\alpha) k(\alpha)} = 0,
    \end{equation*}
    proving \eqref{lemma91b}.
\end{proof}

\begin{lemma}
The following are equivalent:
\begin{subequations}
    \begin{equation}
       \hspace{-3em} \Gamma^{i(\alpha)}_{j(\alpha) k(\alpha)} = 0, \qquad \text{for }   k \geq 2, \,  j\geq i+1;
        \label{lemma92a}
    \end{equation}
    \begin{equation}
        \Gamma^{i(\alpha)}_{j(\alpha)k(\alpha)} = 0, \qquad \text{for } j,k \geq 2, \, i-j-k \leq -3.
        \label{lemma92b}
    \end{equation}
\end{subequations}
    \label{lemma:9.2}
\end{lemma}

\begin{proof}
	Assume that \eqref{lemma92b} holds. Since for $j\geq i+1$ we have $j\geq2$, then
	\begin{equation*}
		k\geq2\,\implies\,i-j-k\leq-3.
	\end{equation*}
	Thus \eqref{lemma92a} holds.
   \newline
   \newline
   For the converse direction, let us assume \eqref{lemma92a}. We proceed by induction over the upper index. Firstly, suppose $i=1$, then the conditions in \eqref{lemma92a} becomes $k\geq 2$, $j\geq i+1 = 2$, and the non-trivial condition in \eqref{lemma92b} reads $i-j-k \leq -3 \iff j+k \geq 4$, which is equivalent to $j,k \geq 2$. Thus $\eqref{lemma92a}$ and \eqref{lemma92b} are equivalent for $i=1$. Now, we want to show that  $\Gamma^{m(\alpha)}_{j(\alpha)h(\alpha)} = 0$ for  $m-h-j \leq -3$ assuming $\Gamma^{t(\alpha)}_{s(\alpha)r(\alpha)} = 0$ for  $t-s-r = m-h-j \leq -3$ with $t<m$. We do so by considering \eqref{nablacsymcoords} for $\alpha = \beta = \gamma = \nu$, with $i=2$, $j,k \geq 2$, and substitute $h \coloneqq k+1$. This gives
   \begin{equation}
       \Gamma^{m(\alpha)}_{j(\alpha)h(\alpha)} = \Gamma^{m(\alpha)}_{2(\alpha) (j+h-2)(\alpha)} - \Gamma^{(m-j+1)(\alpha)}_{2(\alpha)(h-1)(\alpha)} + \Gamma^{(m-1)(\alpha)}_{j(\alpha)(h-1)(\alpha)}.
       \label{eq:lemmaeqtemp}
   \end{equation}
   Letting $m-j-h \leq -3$, we have that the second and third terms of the right-hand side of \eqref{eq:lemmaeqtemp} vanish by the induction hypothesis, and we are left with
    \begin{equation}
       \Gamma^{m(\alpha)}_{j(\alpha) h(\alpha)} = \Gamma^{m(\alpha)}_{2(\alpha) (j+h-2)(\alpha)}. 
        \label{eq:lemmaeqtemp2}
    \end{equation}
    Finally notice that, since $m+1 \leq j+h-2 \iff m-j-h \leq -3$, the right hand side of \eqref{eq:lemmaeqtemp2} vanishes by \eqref{lemma92a} which proves  \eqref{lemma92b}.
\end{proof}

\begin{lemma}
Let $\alpha \neq \beta$.  Then the following are equivalent:
\begin{subequations}
    \begin{equation}
        \Gamma^{i(\alpha)}_{j(\beta)s(\beta)} = 0, \qquad \text{for }  s>i;
        \label{lemma93a}
    \end{equation}
    \begin{equation}
\Gamma^{i(\alpha)}_{1(\beta) k(\beta)} = 0, \qquad \text{for }  k>i;
        \label{lemma93b}
    \end{equation}
	\begin{equation}
		\Gamma^{i(\alpha)}_{1(\alpha) k(\beta)} = 0, \qquad \text{for }  k>i.
		\label{lemma93c}
	\end{equation}
\end{subequations}
    \label{lemma:9.3}
\end{lemma}

\begin{proof}
	Let us first observe that \eqref{lemma93a} implies \eqref{lemma93b}. As for the converse implication, let us assume \eqref{lemma93b}. For $s>i$ and any $j\geq1$ we have
	\begin{align*}
		\Gamma^{i(\alpha)}_{j(\beta)s(\beta)}&\overset{\eqref{lemma91b}}{=}\Gamma^{i(\alpha)}_{1(\beta)(j+s-1)(\beta)}\overset{\eqref{lemma93b}}{=}0.
	\end{align*}
	Finally, by \eqref{Chr_nablae_1}, we have 
   	\begin{equation*}
       \Gamma^{i(\alpha)}_{1(\beta)k(\beta)} = -\Gamma^{i(\alpha)}_{k(\beta) 1(\alpha)},
   	\end{equation*}
   	which implies the equivalence with \eqref{lemma93c}. 
\end{proof}

\subsection{One Jordan block}
Let $V$, as in \eqref{SHT-intro}, be given by
    \begin{equation}
	V=
	\begin{bmatrix}
		X^{1} & 0 & \dots & 0\cr
		X^{2} & X^{1} & \dots & 0\cr
		\vdots & \ddots & \ddots & \vdots\cr
		X^{n} & \dots & X^{2}&  X^{1}
        \label{Vfor1block}
	\end{bmatrix}.
\end{equation}
In this subsection we show that the set of symmetries is complete if $X^i$, for any $i \in \{1, \cdots, n\}$, is a function of $u^1, \cdots, u^i$ only  and that such a requirement is necessary for the applicability of Darboux's Theorem III. This is accomplished by first showing that $X^i$ depends only on $u^1, \cdots, u^i$ if and only if $\Gamma^{i}_{jk}$ is zero whenever $i-j-k \leq -3$.

\begin{lemma}
	The following  are equivalent:
    \begin{subequations}
        \begin{equation}
            \partial_jX^i=0, \quad \forall  i,j\in \{1, \cdots, n\}\,\,\text{such that}\,\,j>i;
            \label{1blockpartialvanish}
        \end{equation}
        \begin{equation}
            \Gamma^i_{jk}=0, \quad\forall i,j,k\in \{1, \cdots, n\}\,\,\text{such that}\,\,i-j-k \leq -3.
            \label{1blockgammavanish}
        \end{equation}
    \end{subequations}		
    \label{prop:1blockgammapartial}
	\end{lemma}

\vspace{-0.8em}
    
\begin{proof}
       Let us first show that \eqref{1blockgammavanish} implies \eqref{1blockpartialvanish}. Considering \eqref{redhie} for $X$, and relabelling $i \leftrightarrow k$, for $k<i$, gives
\begin{equation*}
  	\partial_i X^{k}  = - \sum_{l=2}^{n} \Gamma^k_{il} X^{l}
  	\label{eq:1blockdnabla1}
\end{equation*}
which vanishes due to \eqref{1blockgammavanish}  as $k-i-l \leq  k-i-2 \leq -3$ since $k<i \implies k-i \leq -1$. 
  	\newline
    \newline
   We now prove the converse direction. That is,  we assume that the $i^{\text{th}}$ component of $X$ depends on the first $i$ canonical coordinates only, and show that we must have $\Gamma^{k}_{i l} = 0$ for $k-i-l \leq -3$. 
        By \eqref{Chr_d} we have
        \begin{equation}
            \Gamma^{i}_{jk} = \dfrac{1}{X^2}\bigg(\partial_{j-1}X^{i-k+1} - \partial_k X^{i-j+2} + \sum_{s=k+1}^{n} \Gamma^{i}_{j-1,s} X^{s-k+1} - \sum_{j+1}^{n} \Gamma^{i}_{ks} X^{s-j+2} \bigg).
            \label{eq:gamma1block}
        \end{equation}
        Notice that if $n=1$ we cannot have $i-j-k \leq -3$. Moreover, if $n=2$, the only possibility is $i=1, j=k=2$. However, by \eqref{eq:gamma1block}, $\Gamma^{1}_{22}$ is proportional to $\partial_2 X^1$ which is zero by assumption. Thus we may assume that $n>2$. 
 Since setting either $j$ or $k$ to 1 gives zero by $\nabla e = 0$, it is sufficient to prove \eqref{1blockgammavanish} for $j,k \geq 2$. We shall now prove the proposition by showing that $\Gamma^{i}_{jk} = 0$ whenever $j\geq 2$, and $i<k$ which is equivalent by Lemma \ref{lemma:9.2}. We continue the proof by induction over the lower indices, and without loss of generality we assume that $j \leq k$. Notice that the first two terms in \eqref{eq:gamma1block} vanish as $X^{i-k+1}$ is undefined and $i-j+2<k-j+2 \leq k$ since $j \geq 2$. Hence, 
        \begin{equation}
            \Gamma^{i}_{jk} = \dfrac{1}{X^2}\bigg(\sum_{s=k+1}^{n}\Gamma^{i}_{j-1,s} X^{s-k+1} - \sum_{s=j+1}^{n}\Gamma^{i}_{ks} X^{s-j+2} \bigg).
            \label{eq:1blockgammaexpr}
        \end{equation}
        By letting $j = k = n$, we directly obtain zero by \eqref{eq:1blockgammaexpr}, as the sums are undefined. Thus, let $k=n$ and fix $j \in \{2, \cdots, n\}$, and assume that $\Gamma^i_{hn}=0$, $i<h$, for $h>j$. By \eqref{eq:1blockgammaexpr} we have
        \begin{equation*}
        	\Gamma^i_{j n} = -\dfrac{1}{X^2}\left(\sum_{s=j+1}^{n} \Gamma^i_{ks}  X^{s-j+2}\right),
        \end{equation*}
        which equals zero as each term in the sum vanishes by the inductive assumption. 
        \newline
        \newline
        Now let us fix $k \in \{2, \cdots, n\}$, and assume that $\Gamma^i_{jh} = 0$ for $i<h$, for each $h>k$. By this assumption, \eqref{eq:1blockgammaexpr} reads 
        \begin{equation}
            \Gamma^i_{jk} = - \dfrac{1}{X^2}\sum_{s=j+1}^{n}\Gamma^i_{ks}X^{s-j+2}.
            \label{eq:indtemp1block}
        \end{equation}
        Finally we perform strong induction over $j$ with base $j=k$. From \eqref{eq:indtemp1block} it is clear that setting $j=k$ gives 0 by the above, as we already assumed that $\Gamma^i_{k s} = 0$ for $s>k$. Thus let us fix $j \in \{2, \cdots, k-1\}$ and assume that $\Gamma^i_{hk} = 0$ for $h>j$.  Then, the terms in the sum on the right-hand side of \eqref{eq:indtemp1block} vanish by the inductive assumption, which concludes the proof.
\end{proof}

\begin{lemma}
    Let $j, k \in \{1, \cdots, n\}$. If  $X^i$ is a function of $u^1, \cdots, u^i$ only, for all $i \in \{1, \cdots, n\}$, then   $\Gamma^i_{jk}$ is a function of $u^1, \cdots, u^i$ only, for any $i \in \{1, \cdots, n\}$.  
    \label{lemma:1blockgammadependence}
\end{lemma}
\begin{proof}
Without loss of generality we assume that $j \leq k$. By Lemma \ref{prop:1blockgammapartial}, it is sufficient to consider $\Gamma^i_{jk}$ when $i-j-k \geq -2$, or equivalently $i\geq j, k$ with $j,k \geq 2$ by Lemma \ref{lemma:9.2}. Notice that since $i-k+1 < i$, and $i-j+2 \leq i$, we already have that the first and the third term in \eqref{Chr_d} satisfy the required property. Thus, it remains to consider the $u-$dependence of
\begin{equation}
    \sum_{s=k+1}^{i} \Gamma^{i}_{j-1,s} X^{s-k+1} - \sum_{s=j+1}^{i} \Gamma^{i}_{ks} X^{s-j+2}.  
    \label{eq:1blockint1}
\end{equation}
Furthermore, by Lemma \ref{prop:1blockgammapartial}, $s\leq i-j+3\leq i+1$ in the first sum, implying $s-k+1\leq i-k+2\leq i$, while $s\leq i-k+2\leq i$ in the second sum, implying $s-j+2\leq i-j+2\leq i$. Thus, we only need to prove the required property for the Christoffel symbols involved in \eqref{eq:1blockint1}. Let us proceed by induction over $j+k$. Let $j+k=2n$, which implies $j=k=n$. In this case both the sums in \eqref{eq:1blockint1} are undefined and we get the required result. Thus, assume the statement holds true for any sum of lower indices strictly greater than $j+k$. Here, the only term not directly satisfying the criteria by the inductive assumption is $\Gamma^i_{j-1,k+1}$.  
\newline
\newline
We have now proved that the $u$-dependence of $\Gamma^{i}_{jk}$ follows from the $u$-dependence of $\Gamma^{i}_{j-1, k+1}$. However, this allows us to apply the same argument to $\Gamma^i_{j-1, k+1}$ and so on. There are two possible end-points to this path. The first, which is valid for $k+j-1 \leq n$, gives  $\Gamma^i_{1, k+j-1}$, which is zero by $\nabla e = 0$, and thus trivially satisfies the statement. The second, which is valid for $k+j-1>n$ gives $\Gamma^i_{j-n+k, n}$. This is zero by Lemma \ref{lemma:9.2} unless $i=n$. However, if $i=n$, Lemma \ref{lemma:1blockgammadependence} is automatically satisfied. This concludes the proof.
\end{proof}

\begin{remark}\label{Remark_GammaDependenceMakesCompatChainClosed1block}
	Under the assumption that $\partial_k X^i = 0$ for $k>i$, the formulas for compatibility  \eqref{Compat_temp_1}, \eqref{Compat_temp_2} depend on $u^1, \dots, u^{i}$ only. By using Lemma \ref{prop:1blockgammapartial}, \eqref{Compat_temp_1} gives
    \begin{equation*}
        \partial_j \partial_k Y^i =  -\sum_{s=1}^{i}(\partial_j\Gamma^i_{ks})Y^s - \sum_{s=1}^{i}\Gamma^i_{ks}\partial_j Y^s + \partial_j \nabla_e Y^{i-k+1}, 
    \end{equation*}
    and so the statement follows from the $u$-dependence of the Christoffel symbols (Lemma \ref{lemma:1blockgammadependence}).
    In fact, $R^{i}_{s j k}$ vanishes whenever $s\geq i+1$, and when non-vanishing it depends on $u^1, \dots, u^{i}$ only. This can be seen as follows. The Riemann tensor is given by
    \begin{equation*}
        R^{i}_{sjk} = \partial_k \Gamma^i_{js}-\partial_j \Gamma^i_{ks} + \Gamma^i_{kt} \Gamma^t_{js}- \Gamma^i_{jt} \Gamma^t_{ks}.
    \end{equation*}
    Let $s>i$, then, 
    \begin{itemize}
        \item the first two terms vanish by Lemmas \ref{lemma:9.2}, \ref{prop:1blockgammapartial};
        \item for the third or fourth term  not to vanish we must have $i \geq t$ and $t \geq s$, but this implies $i \geq s$ which is false by assumption. 
    \end{itemize}
    Thus, $R^i_{sjk} = 0$ for any $i,j,k \in \{1, \cdots, n\}$ and $s \in \{i+1, \cdots, n\}$. For $s \leq i$ it is clear from the $u-$dependence of the Christoffel symbols (Lemma \ref{lemma:1blockgammadependence}) that the $u$-dependence of the Riemann tensor is as stated, since the last two terms vanish whenever $t > i$ by Lemmas  \ref{lemma:9.2}, \ref{prop:1blockgammapartial}. 
\end{remark}

We have already pointed out that  in general it  is  not  possible  to apply  Darboux's theorem  to the  system for  the symmetries \eqref{redhie} 
 due  to the presence of partial derivatives of the unknown functions on the  right hand side of the system.  However, we observe that Darboux's theorem can be applied to  the subsystem for $Y^1$ provided  that it is closed, in the sense that it only depends on the arguments of $Y^1$ (in this case, solely on $u^1$) and does not include any of the unknown elements $\{Y^{s}\,|\,s>1\}$. Assuming that the system \eqref{redhie} admits $n$  linearly independent solutions $\{Y_A\}_{A\inw}$, this implies some additional conditions on the Christoffel symbols appearing on the right hand side of \eqref{redhie}. If these conditions are satisfied, the function $Y^1$ is determined up to an arbitrary function of $u^1$. At this stage, the right hand side of the subsystem for $Y^2$ does not contain any partial  derivatives of unknown functions and, reasoning in a similar manner, it can be determined up to an arbitrary function of $u^1$. The same can  be done for all the remaining components of $Y$ leading to the  following results.

\begin{proposition}
   Let $i,j, k \in \{1, \cdots, n\}$. In order to apply Darboux's theorem to the system of symmetries \eqref{redhie}, we must have that  $\Gamma^i_{jk} = 0$ whenever $i-j-k \leq -3$.
   \label{prop:1blockcompletemeansgamma}
\end{proposition}
\vspace{-0.8em}
\begin{proof}
    In the  case of  a single Jordan  block the system \eqref{redhie} reads
    \begin{equation}
        \partial_k Y^i=-\Gamma^i_{ks}Y^s+c^i_{ks}\nabla_eY^s = \partial_1 Y^{i-k+1} - \Gamma^i_{ks} Y^s,\qquad k\geq2.
        \label{eq:1blocksymsyst}
    \end{equation}
      In particular, the subsystem for $i=1$ reads
    \begin{equation}
        \partial_k Y^1 = -\sum_{s=2}^{n}\Gamma^1_{ks}Y^s, \qquad k \geq 2.
        \label{eq:oneblocki1}
    \end{equation}
    Since we want  that the subsystem is closed, the components $Y^s$ for $s>1$ cannot be present in the right-hand side of \eqref{eq:oneblocki1}. By linear independence of the elements of $\{Y_A\}_{A\inw}$, we must have that
    \begin{equation*}
        \Gamma^1_{ks} = 0, \qquad s \geq 2,
    \end{equation*}
    which by Lemma \ref{lemma:9.2} is the desired result for $i=1$. Analogously, as we increase $i$ step by step,  it is clear from \eqref{eq:1blocksymsyst} that in order to have a closed subsystem, we must impose that
    \begin{equation*}
      \sum_{s=2}^n   \Gamma^i_{ks}  Y^s
    \end{equation*}
    does not depend on $Y^s$ for $s>i$. However, by linear independence of the elements of $\{Y_A\}_{A\inw}$, the only way for this to occur is that
    \begin{equation*}
        \sum_{s=i+1}^{n} \Gamma^i_{ks}Y^s = 0 \implies   \Gamma^i_{ks} = 0, \quad \text{for } s>i.
    \end{equation*}
    Thus, by Lemma \ref{lemma:9.2}, we obtained the statement. 
\end{proof}

\begin{proposition}
Let $i,j \in \{1, \cdots, n\}$. If $\partial_i X^j = 0$ for $i>j$, then the set of symmetries is complete. 
\label{prop:1blockgammameanscomplete}
\end{proposition}

\begin{proof}
    The system for the symmetries \eqref{redhie} reads 
\begin{equation*}
	\partial_{k}Y^{i}=\partial_{1}Y^{i-k+1}-\sum^n_{s=2}\,\Gamma^{i}_{ks}\,Y^{s},\qquad i\inw,\,k\in\{2,\dots,n\},
\end{equation*}
which, by Lemma \ref{prop:1blockgammapartial}, amounts to
\begin{equation}\label{EqSymm_sameblock_J1}
	\partial_{k}Y^{i}=\partial_{1}Y^{i-k+1}-\sum^{i-k+2}_{s=2}\,\Gamma^{i}_{ks}\,Y^{s},\qquad i\inw,\,k\in\{2,\dots,n\}.
\end{equation}
In particular, we have
\begin{equation*}\label{EqSymm_sameblock_J1_kgreateri}
	\partial_{k}Y^{i}=0,\qquad k>i.
\end{equation*}
In turn, the system \eqref{EqSymm_sameblock_J1} for $\{Y^i(u^1,\dots,u^n)\}_{i\inw}$ can be rewritten as $n$ systems
\begin{equation}\label{EqSymm_sameblock_J1_kleqi}
	\Bigg\{\partial_{k}Y^{i}=\partial_{1}Y^{i-k+1}-\sum^{i-k+2}_{s=2}\,\Gamma^{i}_{ks}\,Y^{s},\qquad k\in\{2,\dots,i\}\Bigg\}_{i\inw},
\end{equation}
for the functions $\{Y^i(u^1,\dots,u^i)\}_{i\inw}$. When looking at each of these systems separately, Darboux's Theorem III from \cite{darboux} can indeed be applied. More precisely, let us start by observing that for $i=1$ the system \eqref{EqSymm_sameblock_J1_kleqi} is trivially satisfied, determining $Y^1(u^1)$ as an arbitrary function of $u^1$. As for $i=2$, \eqref{EqSymm_sameblock_J1_kleqi} gives
\begin{equation*}\label{EqSymm_sameblock_J1_kleqi_i2}
	\partial_{2}Y^{2}=\partial_{1}Y^{1}-\Gamma^{2}_{22}\,Y^{2},
\end{equation*}
the right-hand side of which only involves the unknown function $Y^2$ and the known quantities $\Gamma^2_{22}$, $\partial_1Y^1$, where $\Gamma^2_{22}$ only depends on $u^1,u^2$ by Lemma \ref{lemma:1blockgammadependence}. It follows that $Y^2(u^1,u^2)$ is uniquely determined, having fixed $Y^1(u^1)$, up to another arbitrary function of $u^1$. Let us fix some $h\leq n$ and inductively assume that, for each $i\leq h-1$, the function $Y^i(u^1,\dots,u^i)$ is uniquely determined up to $i$ arbitrary functions of $u^1$. For $i=h$, the system \eqref{EqSymm_sameblock_J1_kleqi} reads
\begin{equation*}
	\partial_{k}Y^{h}=\partial_{1}Y^{h-k+1}-\sum^{h-k+2}_{s=2}\,\Gamma^{h}_{ks}\,Y^{s},\qquad k\in\{2,\dots,h\},
\end{equation*}
the right-hand side of which only involve the unknown function $Y^h$ and the known quantities $\Gamma^h_{ks}$ for $s\in\{2,\dots,h-k+2\}$, $k\in\{2,\dots,h\}$, $Y^j$ for $j\in\{2,\dots,h-1\}$, and $\partial_1Y^t$ with $t\in\{1,\dots,h-1\}$. Thus, the system for $\{\partial_kY^h\}_{k\in\{2,\dots,h\}}$ only contains the function $Y^h$ itself and known quantities, i.e. the functions $\{Y^j\}_{j\in\{1,\dots,h-1\}}$ which have been determined in the previous steps, and none of the functions $\{Y^j\}_{j\in\{h+1,\dots,n\}}$ which are still to be determined. Moreover, by Lemma \ref{lemma:1blockgammadependence}, the system depends on $u^1,\dots,u^h$ only. The compatibility of such a system follows as a consequence of \eqref{rc-intri}, as pointed out in Remark \ref{RemarkCompatVS3RC} with the aid of Remark \ref{Remark_GammaDependenceMakesCompatChainClosed1block}. By Darboux's theorem, $Y^h(u^1,\dots,u^h)$ is uniquely determined up to one more arbitrary function of $u^1$, as $\partial_1Y^h$ appears nowhere in the system. This proves that the complete solution $\{Y^i(u^1,\dots,u^i)\}_{i\inw}$ is uniquely determined up to $n$ arbitrary functions of $u^1$.
\end{proof}
\vspace{0.5em}
Thus, from Propositions \ref{prop:1blockcompletemeansgamma}, \ref{prop:1blockgammameanscomplete} together with Lemma \ref{prop:1blockgammapartial} and \ref{Remark_GammaDependenceMakesCompatChainClosed1block}, we have the main result of this subsection.
\begin{theorem}
   In order to apply Darboux's theorem to the system of symmetries, \eqref{redhie}, we must have that $X^i$ in \eqref{Vfor1block} depends on $u^1, \cdots, u^i$ only. Furthermore, if $X^i$ depends only on $u^1, \cdots, u^i$ then the set of symmetries is complete.  
\end{theorem}

\subsection{Arbitrary Jordan block structure}\label{SubsectionLauricella}
In this subsection, we treat the general case where the matrix $V$ defining the system \eqref{SHT-intro} consists of $r$ blocks, its generic $\alpha^{\text{th}}$ block being of the form
\eqref{toeplitz}. We show that a necessary condition for Darboux's Theorem III to be applicable to the system for the symmetries is that, for $\alpha\in\{1,\dots,r\}$ and $i\in\{1,\dots,m_\alpha\}$, $X^{i(\alpha)}$ does not depend on ${\{u^{j(\sigma)}\,|\,j>i\}_{\sigma\in\{1,\dots,r\}}}$. Under this requirement, we prove completeness of the symmetries, i.e. that the general solution to the associated system depends on $n$ functions of a single variable, of the form \eqref{narbryfcts_regular}. 
\newline
\newline
By virtue of Lemma \ref{lemma:9.1}, the system for the symmetries (\ref{redhie}) reads
\begin{align}
	&\partial_{k(\alpha)}Y^{i(\alpha)}=\partial_{1(\alpha)}Y^{(i-k+1)(\alpha)}+\overset{m_\alpha}{\underset{s=2}{\sum}}\Big(\Gamma^{(i-k+1)(\alpha)}_{1(\alpha)s(\alpha)}-\Gamma^{i(\alpha)}_{k(\alpha)s(\alpha)}\Big)Y^{s(\alpha)},\quad k\in\{2,\dots,m_\alpha\},\label{EqSymm_sameblock_updated}\\
	&\partial_{j(\beta)}Y^{i(\alpha)}=-\overset{i}{\underset{s=1}{\sum}}\Gamma^{i(\alpha)}_{j(\beta)s(\alpha)}Y^{s(\alpha)}-\overset{m_\beta-j+1}{\underset{s=1}{\sum}}\Gamma^{i(\alpha)}_{j(\beta)s(\beta)}Y^{s(\beta)},\qquad\alpha\neq\beta,\,\,j\in\{1,\dots,m_\beta\},\label{EqSymm_distinctblocks_updated}
\end{align}
for all $\alpha\in\{1,\dots,r\}$ and $i\in\{1,\dots,m_\alpha\}$. Let us denote
\begin{align*}
	\mathcal{U}^i:=\{u^{j(\sigma)}\,|\,j\leq i\}_{\sigma\in\{1,\dots,r\}}=\{u^{1(\sigma)},\dots,u^{\min\{i,m_\sigma\}(\sigma)}\}_{\sigma\in\{1,\dots,r\}}.
\end{align*}
Thus, the condition for $X^{i(\alpha)}$ to not depend on ${\{u^{j(\sigma)}\,|\,j>i\}_{\sigma\in\{1,\dots,r\}}}$ for any ${\alpha\in\{1,\dots,r\}}$ and ${i\in\{1,\dots,m_\alpha\}}$, translates to the requirement that $X^{i(\alpha)}$ depends only on $\mathcal{U}^i$. Under this assumption, completeness of the symmetries follows as a consequence of rearranging the system \eqref{EqSymm_sameblock_updated}-\eqref{EqSymm_distinctblocks_updated} into appropriate subsystems, to each of which Darboux's theorem applies. We begin our discussion by presenting the two following technical results.
\begin{lemma}\label{Proposition_VanishingChrVsX_sameblock}
	Let $\alpha\in\{1,\dots,r\}$. The requirement
	\begin{equation}\label{VanishingByClosedness_sameblock_equiv}
		\Gamma^{i(\alpha)}_{j(\alpha)k(\alpha)}=0,\qquad j,k\geq2,\,\,i-j-k\leq-3,\,\, i\in\{1,\dots,m_\alpha\},
	\end{equation}
	is equivalent to
	\begin{align}\label{XdependenceByClosedness_sameblock}
		\partial_{k(\alpha)}X^{i(\alpha)}=0,\qquad k>i,\,\,i\in\{1,\dots,m_\alpha\}.
	\end{align}
\end{lemma}
\begin{proof}
	Firstly, we assume \eqref{VanishingByClosedness_sameblock_equiv}.  By replacing $Y$ with $X$ in \eqref{EqSymm_sameblock_updated}, for $k>i$ we get
	\begin{align*}
		\partial_{k(\alpha)}X^{i(\alpha)}&=-\overset{m_\alpha}{\underset{s=2}{\sum}}\Gamma^{i(\alpha)}_{k(\alpha)s(\alpha)}Y^{s(\alpha)}\overset{\eqref{VanishingByClosedness_sameblock_equiv}}{=}0.
	\end{align*}
The proof of the converse direction is identical to the one-block case, i.e. Lemma \ref{prop:1blockgammapartial} so we omit the details here. The attentive reader might however notice that in the proof of Lemma \ref{prop:1blockgammapartial} we use the fact that $\nabla e = 0 \implies \Gamma^i_{1j} = 0$, which only holds in the case of a single Jordan block. Nevertheless, this is used solely to remove the $j,k \geq 2$ requirement, which we keep here.
\end{proof}

\begin{lemma}\label{Proposition_VanishingChrVsX_differentblocks}
	Let $\alpha,\beta\in\{1,\dots,r\}$ with $\beta\neq\alpha$. The requirement
	\begin{equation}\label{VanishingByClosedness_differentblocks_equiv2}
		\Gamma^{i(\alpha)}_{1(\alpha)s(\beta)}=0,\qquad s>i,\,\,i\in\{1,\dots,m_\alpha\},
	\end{equation}
	is equivalent to
	\begin{align}\label{XdependenceByClosedness_diffblock}
		\partial_{j(\beta)}X^{i(\alpha)}=0,\qquad j>i,\,\,i\in\{1,\dots,m_\alpha\}.
	\end{align}
\end{lemma}
\begin{proof}
	Let us first assume \eqref{VanishingByClosedness_differentblocks_equiv2}. By choosing $k=1$ in \eqref{Chr_b}, we have
	\begin{align}
		&\Gamma^{i(\alpha)}_{j(\beta)1(\alpha)}\overset{\eqref{lemma91a}}{=}-\frac{1}{X^{1(\alpha)}-X^{1(\beta)}}\bigg(\partial_{j(\beta)}X^{i(\alpha)}+\overset{i}{\underset{s=2}{\sum}}\Gamma^{(i-s+1)(\alpha)}_{j(\beta)1(\alpha)}X^{s(\alpha)}\label{Lemma3.11_k=1}\\
		&\qquad\qquad\quad\,\,\,-\overset{m_\beta}{\underset{s=j+1}{\sum}}\Gamma^{i(\alpha)}_{1(\alpha)s(\beta)}X^{(s-j+1)(\beta)}\bigg),\notag
	\end{align}
	yielding $\partial_{j(\beta)}X^{i(\alpha)}=0$ when $j>i$, since $i-s+1 \leq i-1 <j$ for $s \geq 2$.
	\newline
    \newline
	Let us now assume \eqref{XdependenceByClosedness_diffblock}. Without loss of generality we may let $m_\beta>1$. Condition \eqref{Lemma3.11_k=1} for $i<j=m_\beta$ reads
	\begin{equation}
		\Gamma^{i(\alpha)}_{m_\beta(\beta)1(\alpha)}=-\frac{1}{X^{1(\alpha)}-X^{1(\beta)}}\overset{i}{\underset{s=2}{\sum}}\Gamma^{(i-s+1)(\alpha)}_{m_\beta(\beta)1(\alpha)}X^{s(\alpha)},
        \label{eq:gammaindstep}
	\end{equation}
	which yields $\Gamma^{1(\alpha)}_{m_\beta(\beta)1(\alpha)}=0$ for $i=1$. Thus it is clear from \eqref{eq:gammaindstep} that assuming the statement for any $h < i$ proves that $\Gamma^{i(\alpha)}_{m_{\beta}(\beta)1(\alpha)}=0$ for a generic $i\in\{2,\dots,m_{\beta}-1\}$, by the method of induction. Let us fix now ${j\in\{2,\dots,m_\beta-1\}}$, and inductively assume that $\Gamma^{i(\alpha)}_{h(\beta)1(\alpha)}=0$ for $i<h$, for each $h\in\{j+1,\dots,m_\beta\}$. Then, for $i<j$, condition \eqref{Lemma3.11_k=1} reads
	\begin{equation}
		\Gamma^{i(\alpha)}_{j(\beta)1(\alpha)}=-\frac{1}{X^{1(\alpha)}-X^{1(\beta)}}\overset{i}{\underset{s=2}{\sum}}\Gamma^{(i-s+1)(\alpha)}_{j(\beta)1(\alpha)}X^{s(\alpha)}.
        \label{eq:gammaindstep2}
	\end{equation}
	Performing strong induction over $i$ with base $i=1$ on \eqref{eq:gammaindstep2} concludes the proof.
\end{proof}

\begin{remark}
 Let $X^{s(\rho)}$ depend on $\mathcal{U}^s$ only, for all $\rho \in \{1, \cdots, r\}$ and ${s \in \{1, \cdots, m_{\rho}\}}$. Then 
\begin{equation*}
\Gamma^{i(\alpha)}_{j(\beta) k(\gamma)} = 0, 
\end{equation*}
for any $\alpha, \beta, \gamma \in \{1, \cdots, r\}$, $i \in \{1, \cdots, \max(j,k)-1\}$, $j \in \{1, \cdots, m_{\beta}\}$, ${k \in \{1, \cdots, m_{\gamma}\}}$. This follows from Lemmas \ref{lemma:9.2}, \ref{lemma:9.3},   \ref{Proposition_VanishingChrVsX_sameblock}, and \ref{Proposition_VanishingChrVsX_differentblocks}  together with \eqref{Chr_a}, \eqref{Chr_nablae_2}, and \eqref{lemma91a}.
\label{Rmk:vanishingGammas}
\end{remark}

\begin{lemma}\label{GammaDependence_fromX}
	If, for all $\alpha\in\{1,\dots,r\}$ and $i\in\{1,\dots,m_\alpha\}$, $X^{i(\alpha)}$ depends on $\mathcal{U}^i$ only, then $\Gamma^{i(\alpha)}_{j(\beta)k(\gamma)}$ depends only on $\mathcal{U}^i$ as well, for any $\alpha, \beta, \gamma \in \{1, \cdots, r\}$, $i \in \{1, \cdots, m_{\alpha}\}$, $j \in \{1, \cdots, m_{\beta}\}$, and $k \in \{1, \cdots, m_{\gamma}\}$.
\end{lemma}
\begin{proof}
	The lemma follows from Remark \ref{ChristoffelExplicit}, and we give a sketch-proof in the following. Let us assume that, for all $\alpha\in\{1,\dots,r\}$ and $i\in\{1,\dots,m_\alpha\}$, $X^{i(\alpha)}$ depends on $\mathcal{U}^i$ only. In other words, \eqref{XdependenceByClosedness_sameblock} and \eqref{XdependenceByClosedness_diffblock} hold for all $\alpha\in\{1,\dots,r\}$ and $i\in\{1,\dots,m_\alpha\}$. By \eqref{Chr_a}, the non-trivial Christoffel symbols to be considered are $\{\Gamma^{i(\alpha)}_{j(\beta)k(\gamma)}\}_{\alpha=\beta=\gamma\,\vee\,\gamma=\alpha\neq\beta\,\vee\,\gamma=\beta\neq\alpha}$.
    \newline
    \newline
	Let us first set $\alpha=\beta=\gamma$. Without loss of generality, by Lemma \ref{Proposition_VanishingChrVsX_sameblock}, we consider $i\geq j+k-2$. Moreover, we may assume that $j,k\geq 2$ since the symbols where $j=1$ or $k=1$ are included in the consideration $\gamma = \alpha \neq \beta$ due to \eqref{Chr_nablae_2}. By \eqref{Chr_d} and Lemma \ref{Proposition_VanishingChrVsX_sameblock}, we have
	\begin{align*}
		&\Gamma^{i(\alpha)}_{k(\alpha)j(\alpha)}=\frac{1}{X^{2(\alpha)}}\bigg(\partial_{(j-1)(\alpha)}X^{(i-k+1)(\alpha)}+\overset{i-j+3}{\underset{s=k+1}{\sum}}\Gamma^{i(\alpha)}_{(j-1)(\alpha)s(\alpha)}X^{(s-k+1)(\alpha)}\\&\quad\quad\quad\,\,\,\,-\partial_{k(\alpha)}X^{(i-j+2)(\alpha)}-\overset{i-k+2}{\underset{s=j+1}{\sum}}\Gamma^{i(\alpha)}_{k(\alpha)s(\alpha)}X^{(s-j+2)(\alpha)}\bigg),
	\end{align*}
	where the first and third terms depend on $\mathcal{U}^i$ only, since $i-k+1\leq i$ and $i-j+2\leq i$. As for the second and fourth terms, we have that $\{X^{(s-k+1)(\alpha)}\}_{s\in\{k+1,\dots,i-j+3\}}$ and $\{X^{(s-j+2)(\alpha)}\}_{s\in\{j+1,\dots,i-k+2\}}$ depend on $\mathcal{U}^i$ only, as $s-k+1\leq i-j-k+4\leq i$ and $s-j+2 \leq i-j-k+4 \leq i$ since $j,k \geq 2$. Hence, the property of  $\Gamma^{i(\alpha)}_{k(\alpha)j(\alpha)}$ to depend solely on $\mathcal{U}^i$ reduces to the analogous property of $\{\Gamma^{i(\alpha)}_{(j-1)(\alpha)s(\alpha)}\}_{s\in\{k+1,\dots,i-j+3\}}$ and $\{\Gamma^{i(\alpha)}_{k(\alpha)s(\alpha)}\}_{s\in\{j+1,\dots,i-k+2\}}$, and thus follows from a double inductive argument over $j$ and $k$.
	\newline
    \newline
	Let us now set $\gamma=\alpha\neq\beta$. Without loss of generality, by \eqref{lemma91a}, we consider $i\geq k$. By \eqref{Chr_b}, we have
	\begin{align}\label{eq:gammatemp2}
		\Gamma^{i(\alpha)}_{j(\beta)k(\alpha)}\overset{\eqref{lemma91a}}{\underset{\eqref{VanishingByClosedness_differentblocks_equiv2}}{=}} \, & \, -\frac{1}{X^{1(\alpha)}-X^{1(\beta)}}\Bigg(\partial_{j(\beta)}X^{(i-k+1)(\alpha)}+\overset{i}{\underset{s=k+1}{\sum}}\Gamma^{(i-s+1)(\alpha)}_{j(\beta)1(\alpha)}X^{(s-k+1)(\alpha)}\\ \, & \, -\overset{i-k+1}{\underset{s=j+1}{\sum}}\Gamma^{(i-k+1)(\alpha)}_{1(\alpha)s(\beta)}X^{(s-j+1)(\beta)}\Bigg).\notag
	\end{align}
	The proof that all of the quantities appearing on the right-hand side of \eqref{eq:gammatemp2} depend on $\mathcal{U}^i$ only develops analogously to the previous case.
	\newline
    \newline
	Let us finally set $\gamma=\beta\neq\alpha$. By Lemma \ref{lemma91b}, when $\Gamma^{i(\alpha)}_{k(\beta)j(\beta)}\neq0$, we have
	\begin{align*}
		\Gamma^{i(\alpha)}_{k(\beta)j(\beta)}&=\Gamma^{i(\alpha)}_{1(\beta)(j+k-1)(\beta)}=-\Gamma^{i(\alpha)}_{1(\alpha)(j+k-1)(\beta)},
	\end{align*}
	whose dependence on $\mathcal{U}^i$ only follows from the previous part of the proof.
\end{proof}

\begin{remark}\label{Remark_GammaDependenceMakesCompatChainClosedMultiblock}
	Under the assumptions \eqref{XdependenceByClosedness_sameblock} and \eqref{XdependenceByClosedness_diffblock}, the formulas for compatibility  \eqref{Compat_temp_1}, \eqref{Compat_temp_2} depend on $\mathcal{U}^i$ only when considering $Y^{i(\alpha)}$. This follows from Lemmas \ref{Proposition_VanishingChrVsX_sameblock}, \ref{Proposition_VanishingChrVsX_differentblocks} and \ref{GammaDependence_fromX}, and from the observation that $R^{i(\alpha)}_{s(\sigma) j(\beta) k(\gamma)}$ vanishes when $s>i$, while when non-vanishing it depends on $\mathcal{U}^{i}$ only.
\end{remark}
\textcolor{white}{...}
\vspace{-2.4em}
\newline
\newline
For all $m\in\{1,\dots,\underset{{\sigma\in\{1,\dots,r\}}}{\max}m_\sigma\}$,
let us denote $\mathcal{A}_m:=\{\alpha\in\{1,\dots,r\}\,|\,m_\alpha\geq m\}$.
\begin{proposition}
	In order for the assumptions of Darboux's theorem to be valid for some arrangement of the system \eqref{EqSymm_sameblock_updated}-\eqref{EqSymm_distinctblocks_updated} into suitable subsystems, for all $\alpha\in\{1,\dots,r\}$ and $i\in\{1,\dots,m_\alpha\}$, $X^{i(\alpha)}$ must depend on $\mathcal{U}^i$ only.
\end{proposition}
\begin{proof}
	In order to apply Darboux's theorem to some subsystem of \eqref{EqSymm_sameblock_updated}-\eqref{EqSymm_distinctblocks_updated}, such subsystem must not include terms of the form $\partial_{1(\alpha)}Y^{s(\alpha)}$ for unknown $Y^{s(\alpha)}$. Since such terms appear in the system as $\partial_{1(\alpha)}Y^{(i-k+1)(\alpha)}$ for $k\geq2$, one must start from considering the subsystem for $\{Y^{1(\alpha)}\}_{\alpha\in\{1,\dots,r\}}$, as it is the broadest subsystem where these terms do not appear for any $k\geq2$. It reads
	\begin{align}
		&\partial_{k(\alpha)}Y^{1(\alpha)}=-\overset{m_\alpha}{\underset{s=2}{\sum}}\Gamma^{1(\alpha)}_{k(\alpha)s(\alpha)}Y^{s(\alpha)},\quad k\in\{2,\dots,m_\alpha\},\label{EqSymm_sameblock_updated_1bis}\\
		&\partial_{j(\beta)}Y^{1(\alpha)}=-\Gamma^{1(\alpha)}_{j(\beta)1(\alpha)}(Y^{1(\alpha)}-Y^{1(\beta)})-\overset{m_\beta-j+1}{\underset{s=2}{\sum}}\Gamma^{1(\alpha)}_{j(\beta)s(\beta)}Y^{s(\beta)},\quad\alpha\neq\beta,\,\,j\in\{1,\dots,m_\beta\},\label{EqSymm_distinctblocks_updated_1bis}
	\end{align}
	for all $\alpha\in\{1,\dots,r\}$. In order for the system \eqref{EqSymm_sameblock_updated_1bis}-\eqref{EqSymm_distinctblocks_updated_1bis} to be closed, it must not include any of the elements $\{Y^{s(\alpha)}\,|\,s>1\}_{\alpha\in\{1,\dots,r\}}$, which means we require
	\begin{align*}
		&\overset{m_\alpha}{\underset{s=2}{\sum}}\Gamma^{1(\alpha)}_{k(\alpha)s(\alpha)}Y^{s(\alpha)}=0,\qquad k\geq2,
	\end{align*}
	and
	\begin{align*}
		 &\overset{m_\beta-j+1}{\underset{s=2}{\sum}}\Gamma^{1(\alpha)}_{j(\beta)s(\beta)}Y^{s(\beta)}=0,\qquad j\geq1,\,\,\beta\neq\alpha.
	\end{align*}
	By linear independence of the elements of $\{Y_A\}_{A\inw}$, this implies
	\begin{align*}
		\Gamma^{1(\alpha)}_{k(\alpha)s(\alpha)}=0,\qquad s>1,\,\,k\geq2,
	\end{align*}
	and
	\begin{align*}
		\Gamma^{1(\alpha)}_{j(\beta)s(\beta)}=0,\qquad s>1,\,\,j\geq1,\,\,\beta\neq\alpha.
	\end{align*}
	In this way, one may determine $\{Y^{1(\alpha)}\}_{\alpha\in\{1,\dots,r\}}$. Then, the broadest system to which Darboux's theorem may be applied is the one for $\{Y^{2(\alpha)}\}_{\alpha\in\mathcal{A}_2}$, as the only terms of the form $\partial_{1(\alpha)}Y^{s(\alpha)}$ are the ones for $s=1$, which are known. The same holds at a generic step. More precisely, for $i\in\big\{2,\dots,\underset{{\sigma\in\{1,\dots,r\}}}{\max}m_\sigma\big\}$, the subsystem for $\{Y^{i(\alpha)}\}_{\alpha\in\mathcal{A}_i}$ is such that the only terms of the form $\partial_{1(\alpha)}Y^{s(\alpha)}$ are the ones for ${s\leq i-1}$, which are retrieved in the previous steps. Such system reads \eqref{EqSymm_sameblock_updated}-\eqref{EqSymm_distinctblocks_updated}. In order for it not to include any of the terms $\{Y^{s(\alpha)}\,|\,s>i\}_{\alpha\in\{1,\dots,r\}}$, one must require that
	\begin{align*}
		&\overset{m_\alpha}{\underset{s=i+1}{\sum}}\Big(\Gamma^{(i-k+1)(\alpha)}_{1(\alpha)s(\alpha)}-\Gamma^{i(\alpha)}_{k(\alpha)s(\alpha)}\Big)Y^{s(\alpha)}=0,\qquad k\geq2,
	\end{align*}
	and
	\begin{align*}
		&\overset{m_\beta-j+1}{\underset{s=i+1}{\sum}}\Gamma^{i(\alpha)}_{j(\beta)s(\beta)}Y^{s(\beta)}=0,\qquad j\geq1,\,\,\beta\neq\alpha,
	\end{align*}
	which implies
	\begin{align}\label{CompletenessNecess_i_1}
		\Gamma^{(i-k+1)(\alpha)}_{1(\alpha)s(\alpha)}-\Gamma^{i(\alpha)}_{k(\alpha)s(\alpha)}=0,\qquad s>i,\,\,k\geq2,
	\end{align}
	and
	\begin{align}\label{CompletenessNecess_i_2}
		\Gamma^{i(\alpha)}_{j(\beta)s(\beta)}=0,\qquad s>i,\,\,j\geq1,\,\,\beta\neq\alpha.
	\end{align}
	by linear independence of the elements of $\{Y_A\}_{A\inw}$. In particular, \eqref{CompletenessNecess_i_1} amounts to
	\begin{align}\label{CompletenessNecess_i_1_bis}
		\Gamma^{i(\alpha)}_{k(\alpha)s(\alpha)}=0,\qquad s>i,\,\,k\geq2,
	\end{align}
	as for $s>i$ we have $\Gamma^{(i-k+1)(\alpha)}_{1(\alpha)s(\alpha)}=-\overset{}{\underset{\sigma\neq\alpha}{\sum}}\Gamma^{(i-k+1)(\alpha)}_{1(\sigma)s(\alpha)}=0$ by Lemma \ref{lemma:9.1} as ${s>i>i-k+1}$. By Lemma \ref{lemma:9.2}, \eqref{CompletenessNecess_i_1_bis} amounts to \eqref{VanishingByClosedness_sameblock_equiv}, which is equivalent to \eqref{XdependenceByClosedness_sameblock} by Lemma \ref{Proposition_VanishingChrVsX_sameblock}. By Lemma \ref{lemma:9.3}, \eqref{CompletenessNecess_i_2} amounts to \eqref{VanishingByClosedness_differentblocks_equiv2}, which is equivalent to \eqref{XdependenceByClosedness_diffblock} by Lemma \ref{Proposition_VanishingChrVsX_differentblocks}. By (\ref{XdependenceByClosedness_sameblock}, \ref{XdependenceByClosedness_diffblock}), for all $\alpha\in\{1,\dots,r\}$ and $i\in\{1,\dots,m_\alpha\}$, $X^{i(\alpha)}$ depends on $\mathcal{U}^i$ only.
\end{proof}

\begin{proposition}
	If, for all $\alpha\in\{1,\dots,r\}$ and $i\in\{1,\dots,m_\alpha\}$, the component $X^{i(\alpha)}$ depends on the coordinates $\mathcal{U}^i$ only, then the system \eqref{EqSymm_sameblock_updated}-\eqref{EqSymm_distinctblocks_updated} is complete.
\end{proposition}
\begin{proof}
	Let us assume that, for all $\alpha\in\{1,\dots,r\}$ and $i\in\{1,\dots,m_\alpha\}$, the component $X^{i(\alpha)}$ depends on the coordinates $\mathcal{U}^i$ only, that is
	\begin{align*}
		\partial_{j(\sigma)}X^{i(\alpha)}=0,\qquad j>i,\,\,\sigma\in\{1,\dots,r\}.
	\end{align*}
	By Remark \ref{Rmk:vanishingGammas} and Lemma \ref{lemma:9.2}, the system \eqref{EqSymm_sameblock_updated}-\eqref{EqSymm_distinctblocks_updated} becomes
	\begin{align}
		&\partial_{k(\alpha)}Y^{i(\alpha)}=\partial_{1(\alpha)}Y^{(i-k+1)(\alpha)}+\overset{i-k+2}{\underset{s=2}{\sum}}\Big(\Gamma^{(i-k+1)(\alpha)}_{1(\alpha)s(\alpha)}-\Gamma^{i(\alpha)}_{k(\alpha)s(\alpha)}\Big)Y^{s(\alpha)},\quad k\in\{2,\dots,m_\alpha\},\label{EqSymm_sameblock_updated_underXhp}\\
		&\partial_{j(\beta)}Y^{i(\alpha)}=-\overset{i}{\underset{s=1}{\sum}}\Gamma^{i(\alpha)}_{j(\beta)s(\alpha)}Y^{s(\alpha)}-\overset{i}{\underset{s=1}{\sum}}\Gamma^{i(\alpha)}_{j(\beta)s(\beta)}Y^{s(\beta)},\qquad\alpha\neq\beta,\,\,j\in\{1,\dots,m_\beta\},\label{EqSymm_distinctblocks_updated_underXhp}
	\end{align}
	for all $\alpha\in\{1,\dots,r\}$ and $i\in\{1,\dots,m_\alpha\}$. In particular, \eqref{EqSymm_sameblock_updated_underXhp} trivially vanishes for $k>i$ and, by Remark \ref{Rmk:vanishingGammas}, \eqref{EqSymm_distinctblocks_updated_underXhp} vanishes for $j>i$. Then, a solution $Y$ to \eqref{EqSymm_sameblock_updated_underXhp}-\eqref{EqSymm_distinctblocks_updated_underXhp} must be such that for all $\alpha\in\{1,\dots,r\}$ and $i\in\{1,\dots,m_\alpha\}$, $Y^{i(\alpha)}$ depends on $\mathcal{U}^i$ only. The system \eqref{EqSymm_sameblock_updated_underXhp}-\eqref{EqSymm_distinctblocks_updated_underXhp} reads
	\begin{align}
		&\partial_{k(\alpha)}Y^{i(\alpha)}=\partial_{1(\alpha)}Y^{(i-k+1)(\alpha)}+\overset{i-k+2}{\underset{s=2}{\sum}}\Big(\Gamma^{(i-k+1)(\alpha)}_{1(\alpha)s(\alpha)}-\Gamma^{i(\alpha)}_{k(\alpha)s(\alpha)}\Big)Y^{s(\alpha)},\quad k\in\{2,\dots,i\},\label{EqSymm_sameblock_updated_underXhp_i}\\
		&\partial_{j(\beta)}Y^{i(\alpha)}=-\overset{i}{\underset{s=1}{\sum}}\Gamma^{i(\alpha)}_{j(\beta)s(\alpha)}Y^{s(\alpha)}-\overset{i}{\underset{s=1}{\sum}}\Gamma^{i(\alpha)}_{j(\beta)s(\beta)}Y^{s(\beta)},\qquad\alpha\neq\beta,\,\,j\in\{1,\dots,i\},\label{EqSymm_distinctblocks_updated_underXhp_i}
	\end{align}
	for all $\alpha\in\{1,\dots,r\}$ and $i\in\{1,\dots,m_\alpha\}$. In particular, for each $\alpha\in\{1,\dots,r\}$, $Y^{1(\alpha)}$ only depends on $u^{1(1)},\dots,u^{1(r)}$, and, by Lemma \ref{GammaDependence_fromX}, so does the subsystem
	\begin{align}
		&\partial_{1(\beta)}Y^{1(\alpha)}=-\Gamma^{1(\alpha)}_{1(\beta)1(\alpha)}(Y^{1(\alpha)}-Y^{1(\beta)}),\qquad\beta\neq\alpha,\label{EqSymm_distinctblocks_updated_underXhp_1}
	\end{align}
	to which \eqref{EqSymm_sameblock_updated_underXhp}-\eqref{EqSymm_distinctblocks_updated_underXhp} reduces when letting $i=1$. As Darboux's theorem applies to such a subsystem, whose compatibility follows as a consequence of \eqref{rc-intri}, as pointed out in Remark \ref{RemarkCompatVS3RC} with the aid of Remark \ref{Remark_GammaDependenceMakesCompatChainClosedMultiblock}, its general solution depends on $r$ arbitrary functions of a single variable
	\begin{equation*}
		\big\{f_{1(\alpha)}(u^{1(\alpha)})\big\}_{\alpha\in\{1,\dots,r\}}.
	\end{equation*}
	Once in possession of $\{Y^{1(\alpha)}\}_{\alpha\in\{1,\dots,r\}}$, let us consider the subsystem for $\{Y^{2(\alpha)}\}_{\alpha\in\mathcal{A}_2}$:
	\begin{align}
		&\partial_{2(\alpha)}Y^{2(\alpha)}=\partial_{1(\alpha)}Y^{1(\alpha)}+\Big(\Gamma^{1(\alpha)}_{1(\alpha)2(\alpha)}-\Gamma^{2(\alpha)}_{2(\alpha)2(\alpha)}\Big)Y^{2(\alpha)},\label{EqSymm_sameblock_updated_underXhp_2}\\
		&\partial_{j(\beta)}Y^{2(\alpha)}=-\overset{2}{\underset{s=1}{\sum}}\Gamma^{2(\alpha)}_{j(\beta)s(\alpha)}Y^{s(\alpha)}-\overset{2}{\underset{s=1}{\sum}}\Gamma^{2(\alpha)}_{j(\beta)s(\beta)}Y^{s(\beta)},\qquad\alpha\neq\beta,\,\,j\in\{1,2\},\label{EqSymm_distinctblocks_updated_underXhp_2}
	\end{align}
	for all $\alpha\in\{1,\dots,r\}$. By Lemma \ref{GammaDependence_fromX}, it only depends on $\mathcal{U}^2$. As $\partial_{1(\alpha)}Y^{1(\alpha)}$ is known, Darboux's theorem applies to this subsystem, whose compatibility follows as a consequence of \eqref{rc-intri}. Then its general solution depends on $|\mathcal{A}_2|$ arbitrary functions of a single variable
	\begin{equation*}
		\big\{f_{2(\alpha)}(u^{1(\alpha)})\big\}_{\alpha\in\mathcal{A}_2}.
	\end{equation*}
	Let us assume that, for a fixed $i\in\{2,\dots,\underset{{\sigma\in\{1,\dots,r\}}}{\max}m_\sigma\}$, for each $h\leq i-1$ the generic solution of the subsystem
	\begin{align}
		&\partial_{k(\alpha)}Y^{h(\alpha)}=\partial_{1(\alpha)}Y^{(h-k+1)(\alpha)}+\overset{h-k+2}{\underset{s=2}{\sum}}\Big(\Gamma^{(h-k+1)(\alpha)}_{1(\alpha)s(\alpha)}-\Gamma^{h(\alpha)}_{k(\alpha)s(\alpha)}\Big)Y^{s(\alpha)},\quad k\in\{2,\dots,h\},\label{EqSymm_sameblock_updated_underXhp_h}\\
		&\partial_{j(\beta)}Y^{h(\alpha)}=-\overset{h}{\underset{s=1}{\sum}}\Gamma^{h(\alpha)}_{j(\beta)s(\alpha)}Y^{s(\alpha)}-\overset{h}{\underset{s=1}{\sum}}\Gamma^{h(\alpha)}_{j(\beta)s(\beta)}Y^{s(\beta)},\qquad\alpha\neq\beta,\,\,j\in\{1,\dots,h\},\label{EqSymm_distinctblocks_updated_underXhp_h}
	\end{align}
	for all $\alpha\in\{1,\dots,r\}$, depends on $|\mathcal{A}_h|$ arbitrary functions of a single variable
	\begin{equation*}
		\big\{f_{h(\alpha)}(u^{1(\alpha)})\big\}_{\alpha\in\mathcal{A}_h}.
	\end{equation*}
	Once in possession of $\big\{\{Y^{h(\alpha)}\}_{\alpha\in\mathcal{A}_h}\big\}_{h\in\{1,\dots,i-1\}}$, by Lemma \ref{GammaDependence_fromX}, the subsystem \eqref{EqSymm_sameblock_updated_underXhp_i}-\eqref{EqSymm_distinctblocks_updated_underXhp_i} for $\{Y^{i(\alpha)}\}_{\alpha\in\mathcal{A}_i}$ depends on $\mathcal{U}^i$ only. As $\partial_{1(\alpha)}Y^{(i-k+1)(\alpha)}$ is known for $k\in\{2,\dots,i\}$, Darboux's theorem applies to such a subsystem, whose compatibility follows as a consequence of \eqref{rc-intri}, guaranteeing that its general solution depends on $|\mathcal{A}_i|$ arbitrary functions of a single variable
	\begin{equation*}
		\big\{f_{i(\alpha)}(u^{1(\alpha)})\big\}_{\alpha\in\mathcal{A}_i}.
	\end{equation*}
	Therefore, in total, the general solution to the system \eqref{EqSymm_sameblock_updated_underXhp}-\eqref{EqSymm_distinctblocks_updated_underXhp} depends on
	\begin{align*}
		\overset{\underset{{\sigma\in\{1,\dots,r\}}}{\max}m_\sigma}{\underset{i=1}{\sum}}|\mathcal{A}_i|=\overset{r}{\underset{\sigma=1}{\sum}}m_\sigma=n
	\end{align*}
	arbitrary functions of a single variable
	\begin{equation*}
		\big\{\{f_{i(\alpha)}(u^{1(\alpha)})\}_{\alpha\in\mathcal{A}_i}\big\}_{i\in\big\{1,\dots,\underset{{\sigma\in\{1,\dots,r\}}}{\max}m_\sigma\big\}}=\big\{\{f_{i(\alpha)}(u^{1(\alpha)})\}_{i\in\{1,\dots,m_\alpha\}}\big\}_{\alpha\in\{1,\dots,r\}}.
	\end{equation*}
\end{proof}
We have proved the following, which is the second main result of the paper.
\begin{theorem}\label{Theorem_Completeness_multiblock}
	In order to apply Darboux's theorem to the system for the symmetries \eqref{EqSymm_sameblock_updated}-\eqref{EqSymm_distinctblocks_updated}, we must have that, for all $\alpha\in\{1,\dots,r\}$ and ${i\in\{1,\dots,m_\alpha\}}$, $X^{i(\alpha)}$ depends on $\mathcal{U}^i$ only. Furthermore, if, for all $\alpha\in\{1,\dots,r\}$ and ${i\in\{1,\dots,m_\alpha\}}$, $X^{i(\alpha)}$ depends on $\mathcal{U}^i$ only, then the set of symmetries for the system \eqref{EqSymm_sameblock_updated}-\eqref{EqSymm_distinctblocks_updated} is complete.
\end{theorem}

\begin{corollary}
Let $V$, determining the system \eqref{SHT-intro}, be of the form $V=X\circ$ where $\circ$ is given by  the structure constants \eqref{cijkmulti} and  $X^{i(\alpha)}$ depends on $\mathcal{U}^i$ only, for all $\alpha \in \{1, \cdots, r\}$ and $i \in \{1, \cdots, m_{\alpha}\}$. Then, in a neighbourhood of a point $P=(x_0,t_0)$ where
\begin{equation}\label{trans_cond}
u^{1(\alpha)}_x(P)\ne  0,\qquad \forall\alpha\in\{1,\dots,r\},
\end{equation}
any solution to the system \eqref{SHT-intro} can be obtained by the generalised hodograph method.
\end{corollary}
\begin{proof} \emph{Mutatis mutandis}, the proof is analogous to the diagonal case (Theorem 10  in \cite{ts91}), due to the fact that the system of symmetries satisfies the criteria for the use of Darboux's theorem by Theorem \ref{Theorem_Completeness_multiblock} and that the block-diagonal matrix $M$, defined by \eqref{TsMatrix}, is invertible. Moreover, like in Tsarev's case and precisely in the same way, Darboux's Theorem III needs to  be slightly extended in  order to guarantee the uniqueness of the solution to the system of symmetries specified by its value the curve $u^i(x,t_0)=u^i_0(x)$ satisfying pointwise the transversality condition \eqref{trans_cond}.
\end{proof}  

\subsection{Examples}
The family of F-manifolds with compatible connection and flat unit studied in \cite{LP23} does indeed satisfy the Darboux's  condition on $V$. Such F-manifolds are obtained by considering $V = X \circ$ with ${X=E-a_0\,e}$, where $a_0=\overset{r}{\underset{\alpha=1}{\sum}}\,\epsilon_\alpha\,u^{1(\alpha)}$, and go by the name of Lauricella bi-flat F-manifolds. For instance, in the case of a  single Jordan block of size $3$
 the equation $\text{d}_{\nabla}(Y \circ) = 0$ gives
   \begin{subequations}
       \begin{equation}
         \partial_3 Y^1 = \partial_3 Y^2 = \partial_2 Y^1 = 0, \quad \partial_2 Y^2 = \partial_1 Y^1 + \frac{\epsilon_1}{u^2}Y^2,
       \end{equation}
       \begin{equation}
           \partial_2 Y^3 = \partial_1 Y^2 + \frac{\epsilon_1}{u^2}Y^3 - \frac{\epsilon_1}{(u^2)^2}Y^2, \quad \partial_3 Y^3 = \partial_1 Y^1 + \frac{\epsilon_1}{u^2}Y^2.
       \end{equation}
       \label{eq:exlau3d}
   \end{subequations}
Let us apply the procedure explained in Proposition \ref{prop:1blockgammameanscomplete}. We start from the closed subsystem  for $Y^1$:
\[\partial_2 Y^1 = \partial_3 Y^1 = 0,\]
whose general solution is $Y^1 = f_1(u^1)$. Using this information we obtain a closed
 subsystem for $Y^2$:
 \begin{equation*}
 \partial_3 Y^2 = 0, \quad \partial_2 Y^2 = \partial_1 f_1 + \frac{\epsilon_1}{u^2}Y^2,
       \end{equation*}
 whose general solution is given by
   \begin{equation*}
       Y^1 = f_1(u^1), \qquad Y^2 = \dfrac{u^2 f_1'(u^1)}{1-\epsilon_1} + (u^2)^{\epsilon_1}f_2(u^1), 
   \end{equation*}
where $f_2$ is a new arbitrary function of $u^1$. 
   Substituting $Y^1$ and $Y^2$ in the remaining equations for $Y^3$ we finally get
   \begin{equation*}
   \begin{split}
       Y^3 =  \frac{1}{(1-\epsilon_1)(2- \epsilon_1)} \big(\, & \,(u^2)^2 f_1''(u^1)  + (\epsilon_1-2)\big((u^2)^{\epsilon_1+1}(\epsilon_1-1)f_2'(u^1)- u^3f_1'(u^1) \\ \, & \, + (u^2)^{\epsilon_1-1}(\epsilon_1-1)(u^3f_2(u^1)\epsilon_1 + f_3(u^1)u^2 )\big) \, \big),
       \end{split}
   \end{equation*}
   where $f_1, f_2$, and $f_3$ are arbitrary functions depending only on the main variable $u^1$. Moreover, by Darboux's theorem III, this is the general solution of \eqref{eq:exlau3d}.

As shown in \cite{LPVG}, for general linear functions  $a_0=\overset{n}{\underset{i=1}{\sum}}\,\epsilon_i\,u^{i}$ the construction of  \cite{LP23} still works. However, Darboux's  condition  is no  longer satisfied. For instance, in the case of a  single Jordan block of size $3$, with the  choice $a_0=\epsilon_3u^3$,  the equation $\text{d}_{\nabla}(Y \circ) = 0$ reduces to
\begin{subequations}
\begin{equation}
   \partial_1 Y^1 = \partial_3 Y^3, \quad \partial_1 Y^2 = \partial_2 Y^3, \quad  \partial_2 Y^1 = \partial_3 Y^2 = \dfrac{\epsilon_3}{u^2}\left(\frac{u^3}{u^2}Y^2 - Y^3 \right),
   \label{eq:ex3u31}
   \end{equation}
   \begin{equation}
       \partial_3 Y^1 = -\dfrac{\epsilon_3}{u^2}Y^2, \quad \partial_2 Y^2 = \partial_1 Y^1 - \dfrac{\epsilon_3(u^3)^2}{(u^2)^3}Y^2 + \dfrac{\epsilon_3 u^3}{(u^2)^2} Y^3.
       \label{eq:ex3u32}
   \end{equation}
   \label{eq:ex3u3}
\end{subequations}
Consider the one-form defined by $\alpha \coloneqq \sum_{i=1}^3\alpha_i \text{d}u^i$, where $\alpha_1 = Y^3$, $\alpha_2 = Y^2$, and $\alpha_3 = Y^1$. By \eqref{eq:ex3u31} $\alpha$ is closed and thus there exists locally a function $U(u^1, u^2, u^3)$ such that $\alpha_i = \partial_i U$ for $i = 1, 2, 3$. Rewriting the system \eqref{eq:ex3u3} in terms of $U$ gives
\begin{equation*}
    \partial_3^2 U + \dfrac{\epsilon_3}{u^2}\partial_2 U = 0 = u^2 \partial_2^2 U -  u^2 \partial_1 \partial_3 U + u^3\partial_2\partial_3U,
\end{equation*}
where we substituted the expression for $\partial_2 U$ arising from the third equation of \eqref{eq:ex3u31} into the second of \eqref{eq:ex3u32}. Special  solutions of this system may, for instance, be obtained using the method of separation of variables. However, comparing with the previous example,  we see
 how the method becomes  much less  efficient in  this case. This example shows that when the  first  order linear system for the symmetries does  not satisfy  the hypothesis  of Darboux's theorem III  some higher order linear PDEs come into play.  Based on this and similar examples, we suspect that to further extend  the results of Section \ref{section:completeness} beyond Tsarev-Darboux's  framework might be extremely difficult.

\section{Conclusions and an open problem}
\textcolor{white}{...}
\vspace{-2.4em}
\newline
\newline
In this paper we have introduced the generalised hodograph method for block-diagonal integrable systems of  hydrodynamic type 
\[u^i_t=V^i_j(u)u^j_x,\qquad i\inw,\]
where  $V=\text{diag}(V_{(1)},...,V_{(r)})$ and any block, labelled by $\alpha \in \{1, \cdots, r\}$, has the lower-triangular Toeplitz form 
\begin{equation*}
	V_{(\alpha)}=
	\begin{bmatrix}
		v^{1(\alpha)} & 0 & \dots & 0\cr
		v^{2(\alpha)} & v^{1(\alpha)} & \dots & 0\cr
		\vdots & \ddots & \ddots & \vdots\cr
		v^{m_\alpha(\alpha)} & \dots & v^{2(\alpha)} & v^{1(\alpha)}
	\end{bmatrix}.
\end{equation*} 
These systems can  be  written in the  form
\beq\label{SYS}
{\bf u}_t=X\circ {\bf u}_x,
\eeq
  where  $\circ$ is the commutative associative product of a regular F-manifold. The relation with the theory of F-manifolds was crucial for our purposes as it allowed us to study the problem by framing it in a geometric setting.  
\newline
\newline  
Thanks to the results of \cite{LPVG}, we knew that  integrable systems of the form \eqref{SYS} define on the associated F-manifold a  geometric structure (a torsionless connection $\nabla$) compatible with the product of the F-manifold  structure. In particular,  compatibility implies condition \eqref{shc-intro}. Moreover, according to  the results of \cite{LPR}, in the  semisimple  case and in canonical coordinates, condition \eqref{shc-intro} reduces to Tsarev's integrability condition (systems satisfying  this condition are sometimes called \emph{semihamiltonian systems}, and sometimes \emph{rich systems} \cite{Serre}). This condition coincides with the compatibility condition (in the sense of Darboux) of the linear system providing the symmetries. Due to Darboux compatibility, the general solution of the system depends on $n$ ($n$ being the number of components of the system) arbitrary functions of a single variable. The same condition ensures the existence of a family of densities of conservation laws depending on the same number  of arbitrary functions. 
\newline
\newline
The generalised hodograph method allows the use of any symmetry  of  the system to define a solution which is  given, in implicit form, as $(x,t)$-depending critical points of the vector field 
  \[x\,e+t\,X-Y\]
  obtained as a linear combination of the  unit vector field and  the  vector fields $X$ and  $Y$ defining  the  system and its symmetry, respectively.
\newline
\newline  
The results of the present paper show that the above integration method carries forth to regular non-semisimple hydrodynamic systems. The main difference between the semisimple and the non-semisimple case is the fact that in the latter the  integrability condition \eqref{shc-intro} is not sufficient to guarantee that  the solutions obtained applying the method provide the general solution  of the system.  The reason being that, in general, the  linear system for the symmetries cannot be reduced to Darboux's form. However, Tsarev's theory  can be fully extended considering the additional restrictions arising from  the request of reducibility to Darboux's form. In particular,  we have shown that  the  possibility of applying Darboux's theorem on compatible systems relies on the dependence of the components of the vector fields defining  the system  on a  subset  of the variables. 
\newline
\newline
In a forthcoming paper we will study the Hamiltonian formalism  for the systems considered  in this  paper. For a general system of hydrodynamic type, the metric defining a Hamiltonian structure for the system satisfies the Dubrovin-Novikov-Tsarev conditions  \cite{DN,ts91}
\begin{eqnarray*}
g_{is}V^s_j&=&g_{js}V^s_i;\\
(d_{\nabla_g}V)^i_{jk}&=&0,
\end{eqnarray*}
where $\nabla_g$ is the Levi-Civita connection of $g$. In Tsarev's  diagonal integrable case, the above system admits a  family of solutions (depending on  functional parameters) thanks to Darboux's  theorem III. Some preliminary results suggest that a similar result holds true also in the  regular non-diagonalisable case.
\newline
\newline
{\bf  Conflict of interest}.  The authors have no conflict of  interest to declare that is relevant to the contents of this article.

\end{document}